\documentclass[11pt]{article}

\usepackage{natbib}
\usepackage{amsmath,amsfonts,amsthm}
\usepackage{graphicx}
\usepackage{dsfont} 
\usepackage{enumerate}
\usepackage[titletoc,title]{appendix}
\usepackage[left=3cm,right=3cm,top=3cm,bottom=3cm]{geometry}

\newcommand*{\Scale}[2][4]{\scalebox{#1}{$#2$}}%
\graphicspath{{figures/}}

\newtheorem{theorem}{Theorem}
\newtheorem{proposition}[theorem]{Proposition}
\newtheorem{lemma}[theorem]{Lemma}
\newtheorem{remark}[theorem]{Remark}

\numberwithin{theorem}{section}

\DeclareMathOperator*{\argmin}{arg\,min}

\title{
  A (tight) upper bound for the length of confidence intervals with conditional coverage
}

\author{Danijel Kivaranovic \quad Hannes Leeb\\
Department of Statistics and Operations Research \\
University of Vienna
}
\date{}

\begin{document}
\maketitle

\sloppy

\begin{abstract}

We show that two popular selective inference procedures, namely
data carving \citep{fithian2017} and selection with a randomized 
response \citep{tian2018}, when combined with the polyhedral 
method \citep{lee2016},
result in confidence intervals whose length
is bounded. This contrasts results for confidence intervals based on
the polyhedral method alone, whose expected length is typically
infinite  \citep{kivaranovic2020}. Moreover, we show that
these two procedures always dominate
corresponding sample-splitting methods in terms 
of interval length.
\end{abstract}

\section{Introduction}
\label{sec_intro}

Post-model-selection inference, i.e.,
parametric inference when the fitted model is chosen
in a data-driven fashion, is non-trivial. Obviously, such a model
is random and may be misspecified. It is well-known that
the `naive' approach, where the model-selection step is ignored in the
sense that the chosen model is treated as a-priori given and as correct in
subsequent analyses, 
can result in invalid inference procedures; cf. \cite{Lee03a}.
The polyhedral method of \cite{lee2016} is a recently proposed technique
that allows one to construct valid inference procedures, like tests or
confidence intervals, after model selection, for a parameter of interest
that depends on the selected model.
The polyhedral method
and its variants 
can be used with a variety of methods, including the Lasso or
the sequential testing method considered later in this paper.

\cite{kivaranovic2020} showed that the expected length of confidence 
intervals based on the polyhedral method of \cite{lee2016} is typically
infinite. 
The polyhedral method can be modified by combining it with
data carving \citep{fithian2017} or with selection on a randomized
response \citep{tian2018}. These references found,
in simulations, that this combination results in significantly shorter 
intervals than those based on the polyhedral method alone.
In this paper, we give a formal analysis of this phenomenon.
We show that the polyhedral method, when combined with the proposals of
\cite{fithian2017} or of \cite{tian2018}, delivers intervals whose
length is always bounded.  Our upper bound is easy to
compute, easy to interpret 
and also applies in situations where variances are estimated.
In the interesting case where the polyhedral method alone
gives intervals with infinite expected length, our bound is also sharp.
Moreover, we show that the intervals of \cite{fithian2017}
and \cite{tian2018} are always shorter than the intervals
obtained by corresponding sample splitting methods.

There are several ongoing developments regarding
inference after model 
selection. Roughly speaking, one can divide them into two branches: Inference 
conditional on the selected model and inference simultaneous over all potential 
models. Pioneer works in these two areas are \cite{lee2016} and 
\cite{berk2013}, respectively. 
This paper is concerned with procedures in 
the first branch, i.e., procedures that evolved from 
the polyhedral method. See, among others, 
\cite{fithian2017,	
markovic2018, 
panigrahi2018,	
panigrahi2019, 
reid2017,	 
reid2018,
taylor2018,
tian2016,
tian2017,
tian2018b,
tian2018,
tibshirani2016}. For related literature from the second branch,
see among others, 
\cite{bachoc2019,bachoc2020,kuchibhotla2018a,kuchibhotla2018b,zrnic2020}. 
We also want to 
note that all these works address model-dependent targets and not the true 
underlying parameter in the classical sense. Valid inference for the 
underlying 
true parameter is a more challenging task, as demonstrated by the impossibility 
results of \cite{Lee03a,leeb2006,leeb2008}.

This paper is organized as follows. In Section~\ref{sec_overview},
we showcase our results in the context of selective
inference for the file-drawer problem and with the Lasso, and we
describe a particular
conditional distribution that occurs in selective inference 
when the polyhedral method is combined with data carving or
with selection with a randomized response.
In Section~\ref{sec_core} we present our 
main technical result, Theorem~\ref{th_core}, where we show that a confidence 
interval based on the aforementioned conditional distribution 
has bounded length. We also perform simulations that provide additional
insights on the accuracy of our upper bound. 
In Section~\ref{sec_selective}, we demonstrate that the conditional 
distribution
considered in Theorem~\ref{th_core} frequently arises in selective inference.
In particular, we consider the polyhedral method combined with
data carving as well as the polyhedral method
combined with selection on a randomized response.
We show that these procedures give 
confidence intervals whose length is always bounded, and that they
strictly dominate corresponding sample-splitting methods
in terms of interval length.
While the discussion in Section~\ref{sec_selective} is generic,
we also provide detailed examples in Section~\ref{sec_examples}, 
namely a simple case of sequential testing with
data carving and
model selection using the Lasso on a randomized response.
A discussion in Section~\ref{sec_discussion} concludes. 
All proofs are collected in the appendix.

\section{Overview and Motivation}
\label{sec_overview}

Consider first a very simple scenario for
selective inference, the file-drawer problem, which was 
introduced by \cite{Ros79a} as a simple model for publication bias.
Given $Z \sim N(\mu,1)$,
the goal is to construct a confidence interval for $\mu$ provided, e.g., 
that $Z > 1.64$, that is, provided that a test of the hypothesis $H_0: \mu=0$
is rejected.  
The usual equal-tailed 95\%-confidence interval for $\mu$, i.e., 
$Z\pm 1.96$, does not have the right coverage probability
conditional on the event that $Z > 1.64$ because,
conditional on this event, the data is not Gaussian but distributed as
\begin{align}\label{truncated}
	Z ~ | ~ Z \in T
\end{align}
for $T=(1.64,\infty)$, which is a truncated normal.
This distribution depends on unknown parameters only through 
$\mu$, and standard methods give
an equal-tailed  confidence set $[\hat{L}(Z), \hat{U}(Z)]$ that satisfies
$$
\mathbb P( \mu \in [\hat{L}(Z),\hat{U}(Z)] ~|~ Z \in T) \quad=\quad 0.95;
$$
see \citet[Theorem 5.2]{lee2016}.
This interval, 
while having the right coverage probability,
can be rather large in practice:
Because the truncation set $T$ is bounded on one side, 
the results of~\cite{kivaranovic2020} entail that
$$
\mathbb E( \hat{U}(Z) - \hat{L}(Z) ~|~ Z\in T) \quad=\quad \infty.
$$
Now suppose that the conditioning event is randomized as
in Example 2 of~\cite{tian2018}. More precisely, suppose that
the event $\{Z \in T\}$ is replaced by $\{Z + R \in T\}$, where
$R\sim N(0,\tau^2)$ is independent of $Z$ and where $\tau$ is known.
While this may  not be practical when dealing with publication bias,
this is exactly what happens in 
data carving  and selection with a randomized 
response; see Section~\ref{sec_selective}.
Now the underlying distribution becomes
\begin{align}\label{truncatedRandomized}
	Z~|~Z+R\in T,
\end{align}
and standard methods can again be used to construct an equal-tailed confidence
set $[\check{L}(Z), \check{U}(Z)]$ that satisfies
$$
\mathbb P( \mu \in [\check{L}(Z),\check{U}(Z)] ~|~ Z +R \in T) \quad=\quad 0.95.
$$
The main finding of this paper is that this and related
randomization methods have a dramatic impact on confidence interval length. 
Without randomization, the expected length of the interval 
$[\hat{L}(Y), \hat{U}(Y)]$ is infinite. With randomization,
the length of the interval
$[\check{L}(Y), \check{U}(Y)]$ is always bounded:
$$
\check{U}(Z) -\check{L}(Z) <  2\times 1.96 \times \sqrt{1+1/\tau^2};
$$
see Theorem~\ref{th_core}.
The upper bound is just the length of the usual (unconditional) 95\%-confidence
interval $Z\pm 1.96$ multiplied by $\sqrt{1+1/\tau^2}$.

Now consider a more elaborate scenario, model selection with the Lasso.
Consider a response vector $Y$ and a matrix $X$ of explanatory variables.
In particular,
assume that $Y\sim N(\theta, \sigma^2 I_n)$ with $n\in\mathbb N$, 
$\theta \in \mathbb R^n$ and  $\sigma^2 \in (0,\infty)$, and
assume that
$X \in \mathbb R^{n\times d}$ ($d\in\mathbb N$) is a fixed matrix whose columns
are in general position in the sense of \cite{tibshirani2013}.
The Lasso estimator, denoted by $\hat{\beta}(y)$, is the minimizer of the
least squares problem with an additional penalty on the absolute size
of the regression coefficients \citep{frank1993, tibshirani1996}:
$$
\hat{\beta}(y)\quad=\quad \argmin_{\beta \in \mathbb R^d}
	\frac{1}{2}\| y-X \beta\|_2^2 + \lambda \|\beta\|_1,
$$
where $y\in\mathbb R^n$ and where $\lambda \in (0,\infty)$ is a 
given tuning parameter.
Because of our assumptions on $X$, $\hat{\beta}(y)$
is well-defined; cf. Lemma~3 in \cite{tibshirani2013}.
Since individual components of $\hat{\beta}(Y)$ are 
zero with positive probability,
the non-zero coefficients of $\hat{\beta}(Y)$ can be viewed
as the model $\hat{m}(Y)$ selected by the Lasso.
More formally, for each $y\in \mathbb R^n$,
let $\hat{m}(y) \subseteq \{1,\dots, d\}$  and
$\hat{s}(y)\in \{-1,1\}^{|\hat{m}(y)|}$ 
denote the set of  indices and the vector of signs, respectively, of the
non-zero components of $\hat{\beta}(y)$
(in case $\hat{m}(y) = \emptyset$, $\hat{s}(y)$ is left undefined).

Conditional on events like $\{\hat{m}(Y)=m\}$ or 
$\{\hat{m}(Y)=m, \hat{s}(Y)=s\}$,
the polyhedral method provides confidence intervals for
linear contrasts of the form $\eta_m'\theta$ with pre-specified coverage 
probability  for a given $d$-vector $\eta_m \neq 0$.\footnote{
	The quantity of interest $\eta_m'\theta$ may depend on the
	selected model and is often of the form
	$\eta_m'\theta = \gamma_m' [ (X_m'X_m)^{-1} X_m' \theta]$
	if $m\neq\emptyset$.
}
These intervals are based on the (conditional) distribution of 
$\eta_m'Y$, which is the obvious (unconditionally unbiased) estimator
for $\eta_m'\theta$.
Consider a model $m\neq \emptyset$ and a sign-vector 
$s\in \{-1,1\}^{|m|}$, so that $\mathbb P(\hat{m}(Y)=m,\hat{s}(Y)=s)>0$.
\cite{lee2016} show that the event $\{\hat{m}(Y)=m,\hat{s}(Y)=s\}$
is a polyhedron in $Y$-space. Hence, the distribution of 
$Y | \hat{m}(Y)=m, \hat{s}(Y)=s$ is a multivariate Gaussian restricted to said
polyhedron. Now decompose $Y$ into the sum of two independent 
components were one depends only on 
$\eta_m'Y$: $Y = P_{\eta_m}Y + (I_n-P_{\eta_m})Y$ with
$P_{\eta_m}$ denoting the orthogonal projection on the span of $\eta_m$.
Conditional on $\hat{m}(Y)=m$, $\hat{s}(Y)=s$ {\em and} $(I_n-P_{\eta_m})Y=w$,
we see that $Y$ equals $P_{\eta_m}Y + w$ and lies on the affine line 
$\{\alpha \eta_m + w, \alpha \in \mathbb R\}$ intersected 
with the polyhedron corresponding to $\{\hat{m}(Y)=m, \hat{s}(Y)=s\}$.
In particular, the conditional distribution of $\eta_m'Y$
conditional on 
$\hat{m}(Y)=m$, $\hat{s}(Y)=s$ and $(I_n-P_{\eta_m})Y=w$
is a a truncated normal, similar to \eqref{truncated},
where the truncation set is now an interval
$T_{m,s}(w)$  that depends on the polyhedron
(i.e., on $m$ and $s$) and on $w$.
In particular,
$$
\mathcal L( \eta_m'Y ~|~ \hat{m}(Y)=m,\hat{s}(Y)=s,(I_n-P_{\eta_m})Y=w)
\quad=\quad \mathcal L(Z ~|~ Z \in T_{m,s}(w)),
$$
where $\mathcal L(\dots)$ denotes the indicated (conditional) distributions,
where $Z \sim N(\eta_m'\theta, \sigma^2 \|\eta_m\|^2)$, and where 
$T_{m,s}(w)$ is as above.
Note that this distribution is of the same functional form as 
\eqref{truncated}.
Using this distribution, \cite{lee2016} obtain a confidence interval
for $\eta_m'\theta$ with pre-specified coverage probability conditional
on $\hat{m}(Y)=m,\hat{s}(Y)=s, (I_n-P_{\eta_m})Y=w$
and hence also conditional on the (larger) event
$\{\hat{m}(Y)=m,\hat{s}(Y)=s\}$. The construction is the same as that used
to obtain the interval $[\hat{L}(Z), \hat{U}(Z)]$ in the file-drawer problem
considered earlier. And, again similar to the file-drawer problem, 
the conditional
expected length of this interval is infinite because the truncation set
$T_{m,s}(w)$ is always bounded either from above or below
\citep[Proposition 2]{kivaranovic2020}.

Next, suppose that the model selection step is randomized as proposed
by~\cite{tian2018}. Take $\omega \sim N(0,\tau^2 I_n)$ 
independent of $Y$. We will select a model based on the
randomized data $Y+\omega$, while
the original data $Y$ will be used for subsequent inference.
For model selection, we first 
compute the lasso-estimator $\hat{\beta}(Y+\omega)$
from the randomized data. The non-zero coefficients of this estimator
and their signs again give a model $\hat{m}(Y+\omega)$ and
a sign-vector $\hat{s}(Y+\omega)$ (provided that $\hat{m}(Y+\omega)\neq 0$).
For $m$ and $s$ as before, 
we will use the polyhedral method to obtain a confidence interval
for $\eta_m'\theta$ with pre-specified coverage probability
conditional on the event $\{\hat{m}(Y+\omega)=m,\hat{s}(Y+\omega)=s\}$
that is based on the conditional distribution of $\eta_m'Y$.
The event $\{\hat{m}(Y+\omega)=m, \hat{s}(Y+\omega)=s\}$
is again a polyhedron, but now in $(Y+\omega)$-space.
Arguing as in the preceding paragraph, we see that 
$$
\mathcal L( \eta_m'(Y+\omega) ~|~ \hat{m}(Y+\omega)=m,
	\hat{s}(Y+\omega)=s,(I_n-P_{\eta_m})(Y+\omega)=w)
\quad=\quad \mathcal L(Z +R ~|~ Z +R \in T_{m,s}(w)),
$$
where $Z$ and $T_{m,s}(w)$ are as before and where
$R \sim N(0,\tau^2\|\eta_m\|^2)$ is independent of $Z$.
Because we use the estimator $\eta_m'Y$, the conditional distribution of
interest is
$$
\mathcal L( \eta_m'Y ~|~ \hat{m}(Y+\omega)=m,
	\hat{s}(Y+\omega)=s,(I_n-P_{\eta_m})(Y+\omega)=w)
\quad=\quad \mathcal L(Z ~|~ Z +R \in T_{m,s}(w)).
$$
Based on this distribution, \cite{tian2018} obtain a confidence interval
for $\eta_m'\theta$ with pre-specified coverage probability conditional
on $\hat{m}(Y+\omega)=m, \hat{s}(Y+\omega)=s,(I_n-P_{\eta_m})(Y+\omega)=w$
and hence also conditional on the (larger) event
$\{\hat{m}(Y+\omega)=m, \hat{s}(Y+\omega)=s\}$. 
Because the distribution in the preceding display is of the same 
functional form as \eqref{truncatedRandomized}, the confidence interval
of \cite{tian2018} has properties similar to those of the interval
$[\check{L}(Z), \check{U}(Z)]$ that we constructed in the
randomized file-drawer problem. In particular, its length is
always bounded. We will return to the Lasso later in 
Subsection~\ref{subsec_lasso} to cover some more technical details;
in particular, we will
explicitly compute the upper bound on confidence interval length
and also consider the case where the conditioning is on the
selected model only (and not on the signs).

\section{Main technical result} \label{sec_core}

Here, we study confidence intervals based on an observation
from the truncated Gaussian distribution 
\begin{equation} \label{cond_rv}
  Z ~ | ~ Z + R \in T,
\end{equation}
where $Z$ and $R$ are independent, $Z\sim N(\mu,\sigma^2)$
with $\mu\in\mathbb R$ and $\sigma^2 \in (0,\infty)$,
$R\sim N(0,\tau^2)$ with $\tau^2 \in (0,\infty)$,
and where the truncation set $T$ is of the form
\begin{equation} \label{trunc_set}
  T ~ = ~ \bigcup_{i=1}^k ~ (a_i,b_i)
\end{equation}
with $k \in \mathbb N$ and $\infty \leq a_1 < b_1 < \dots < a_k < b_k \leq 
\infty$. 

Let $\Phi(t)$ be the cumulative distribution function (c.d.f.) 
of the standard normal distribution and denote by 
$F_{\mu,\sigma^2}^{\tau^2}(z)$ the conditional c.d.f. of the random variable in 
\eqref{cond_rv}. For sake of readability, we do not show the dependence of this 
c.d.f. on $T$ in the notation. For $\alpha \in (0,1)$
and $z\in \mathbb R$, let $\mu_\alpha(z)$ satisfy
 \begin{equation} \label{eq_bound}
  F_{\mu_\alpha(z),\sigma^2}^{\tau^2}(z) ~ = ~ 1-\alpha.
 \end{equation}
The quantity $\mu_\alpha(z)$ is well-defined and 
strictly increasing as a function of $\alpha$ for 
fixed $z \in \mathbb R$ (cf. Lemma~\ref{le_cdf}). For all 
$\alpha_1,\alpha_2 \in (0,1)$ such that $\alpha_1 < \alpha_2$, we have
\begin{equation}\label{eq_coverage}
  \mathbb P \left(
  \mu \in [\mu_{\alpha_1}(Z),\mu_{\alpha_2}(Z)] ~ | ~ Z + R \in T
  \right) ~ = ~ \alpha_2 - \alpha_1
\end{equation}
by a textbook result for confidence bounds (e.g., Chapter 3.5 in 
\citealp{lehmann2006}). A common choice is to set $\alpha_1=\alpha/2$ and 
$\alpha_2=1-\alpha/2$ such that, conditional on $\{Z + R \in T\}$, 
$[\mu_{\alpha_1}(Z),\mu_{\alpha_2}(Z)]$ is an equal-tailed confidence interval for $\mu$ 
at level $1-\alpha$. Another option is to choose $\alpha_1$ and $\alpha_2$ such that, 
conditional on $\{Z + R \in T\}$, $[\mu_{\alpha_1}(Z),\mu_{\alpha_2}(Z)]$ is an unbiased 
confidence interval at level $1-\alpha$ (cf. Chapter 5.5 in 
\citealp{lehmann2006}).

It is easy to see that, as $\tau^2$ goes to $0$, $F_{\mu,\sigma^2}^{\tau^2}(z)$ 
converges weakly to the c.d.f. of the truncated normal distribution 
$Z|Z\in T$, which leads to confidence intervals with
infinite expected length if (and only if)
the truncation set $T$ is bounded from above
or from below \citep[Proposition 1]{kivaranovic2020}.
On the other hand, as $\tau^2$ goes to $\infty$, it is similarly easy to see
that $F_{\mu,\sigma^2}^{\tau^2}(z)$ converges weakly to the c.d.f. of the 
normal distribution with mean $\mu$ and variance $\sigma^2$. Hence, in the case 
where the intervals are based on $\lim_{\tau^2\to \infty} 
F_{\mu,\sigma^2}^{\tau^2}(z)$, the length of 
the resulting confidence interval equals
$\sigma \left(\Phi^{-1}(\alpha_2) - \Phi^{-1}(\alpha_1) \right)$. These 
observations suggest that, in the case where $\tau^2 \in (0,\infty)$,  the 
length of $[\mu_{\alpha_1}(z),\mu_{\alpha_2}(z)]$ might be
bounded somewhere between $\sigma 
\left(\Phi^{-1}(\alpha_2) - \Phi^{-1}(\alpha_1) \right)$ and $\infty$.
This idea is formalized by the following result.

\begin{theorem} \label{th_core}
Fix $\sigma^2$ and $\tau^2$ in $(0,\infty)$, $\mu\in\mathbb R$,
$T$  as in \eqref{trunc_set}, and $0<\alpha_1< \alpha_2<1$. For each $x\in \mathbb R$,
we then have
  \begin{equation*}
    \mu_{\alpha_2}(z) - \mu_{\alpha_1}(z) ~ < ~ \sigma \left(\Phi^{-1}(\alpha_2) 
- \Phi^{-1}(\alpha_1) \right) \sqrt{1+\frac{\sigma^2}{\tau^2}},
  \end{equation*}
where
$\mu_{\alpha_i}(z)$ is defined as in \eqref{eq_bound} with $\alpha_i$ replacing $q$,
$i=1,2$.
If $\sup T = b_k < \infty$, the left-hand side  converges to the 
right-hand side as $z \to 
\infty$. The same is true if $\inf T = a_1 > -\infty$ and $z \to -\infty$.
\end{theorem}

The upper bound in Theorem~\ref{th_core} is easy to compute, 
does not depend on the truncation set $T$, and increases as  
the amount of randomization $\tau^2$ decreases.
As $\tau^2$ goes to zero, the upper bound diverges to infinity,
in accordance with \cite{kivaranovic2020}.
On the other hand, as $\tau^2$ goes to $\infty$, the upper bound converges to 
$\sigma \left(\Phi^{-1}(\alpha_2) - \Phi^{-1}(\alpha_1) \right)$. 
Also note that the upper bound is sharp if $T$ is bounded from above
or from below, i.e., in the
case where confidence intervals based on $Z|Z\in T$ have infinite
expected length.
Finally, as detailed in Remark~\ref{remark_variance} below, the upper bound
can also be used in the unknown-variance case, i.e., if $\sigma^2$ or $\tau^2$
or both are replaced by estimators.

In Figure~\ref{fig_length}, we plot the length of 
$[\mu_{\alpha_1}(z),\mu_{\alpha_2}(z)]$ as a function of $z$ for several truncation sets 
$T$. In the left panel the truncation set is of the form $(-a, a)$ (bounded)
and in the 
right panel the truncation set is of the form $(-\infty, -a) \cup 
(a,\infty)$ (unbounded with a gap in the middle). 
The top dashed line denotes 
the upper bound $\sigma
\left(\Phi^{-1}(\alpha_2) - \Phi^{-1}(\alpha_1) \right)\sqrt{1+\sigma^2/\tau^2}$;
the bottom dashed line 
denotes $\sigma\left(\Phi^{-1}(\alpha_2) - \Phi^{-1}(\alpha_1) 
\right)$, i.e., the length of the confidence interval with 
unconditional coverage. In the left panel, we see that $\mu_{\alpha_2}(z) - 
\mu_{\alpha_1}(z)$ approximates the upper bound as $|z|$ diverges. The smaller $a$,
i.e., the smaller the truncation set $T$, the faster the convergence. 
Also in this case, where the truncation set is a bounded
interval, the left panel 
indicates that the length is 
bounded from below by $\sigma\left(\Phi^{-1}(\alpha_2) 
- \Phi^{-1}(\alpha_1) \right)$. In the right panel, we see that our upper bound is 
not sharp when the truncation set is unbounded on both sides. It appears that
the length converges to 
$\sigma\left(\Phi^{-1}(\alpha_2) - \Phi^{-1}(\alpha_1) \right)$ as $z$ diverges. However, 
as the gap of the truncation set becomes larger (i.e.,  as $a$ grows), we 
see that the length approximates the upper bound for values around $a$ and 
$-a$. Finally, $\sigma\left(\Phi^{-1}(\alpha_2) - \Phi^{-1}(\alpha_1) \right)$ is 
not an lower bound in this case, as we can see that the length is considerably 
smaller for values in the gap $(-a,a)$. It seems that the length 
converges to zero for values of $z$ around $0$ as $a$ diverges.

\begin{figure} 
\centering
\begin{minipage}{.5\textwidth}
  \centering
  \centerline{\includegraphics[trim=0 10 0 0, clip, 
width=1\textwidth]{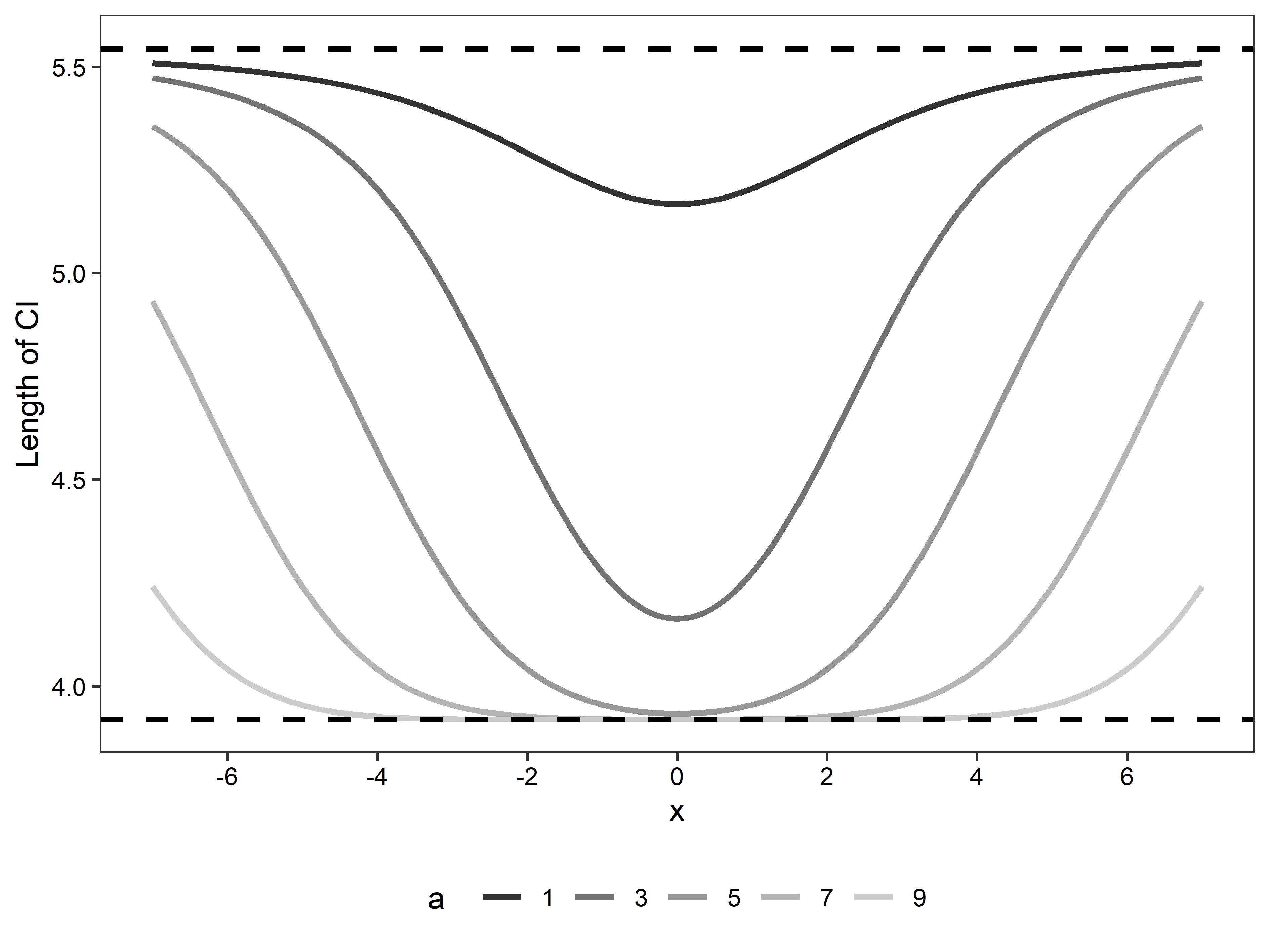}}
\end{minipage}%
\begin{minipage}{.5\textwidth}
  \centering
  \centerline{\includegraphics[trim=0 10 0 0, clip, 
width=1\textwidth]{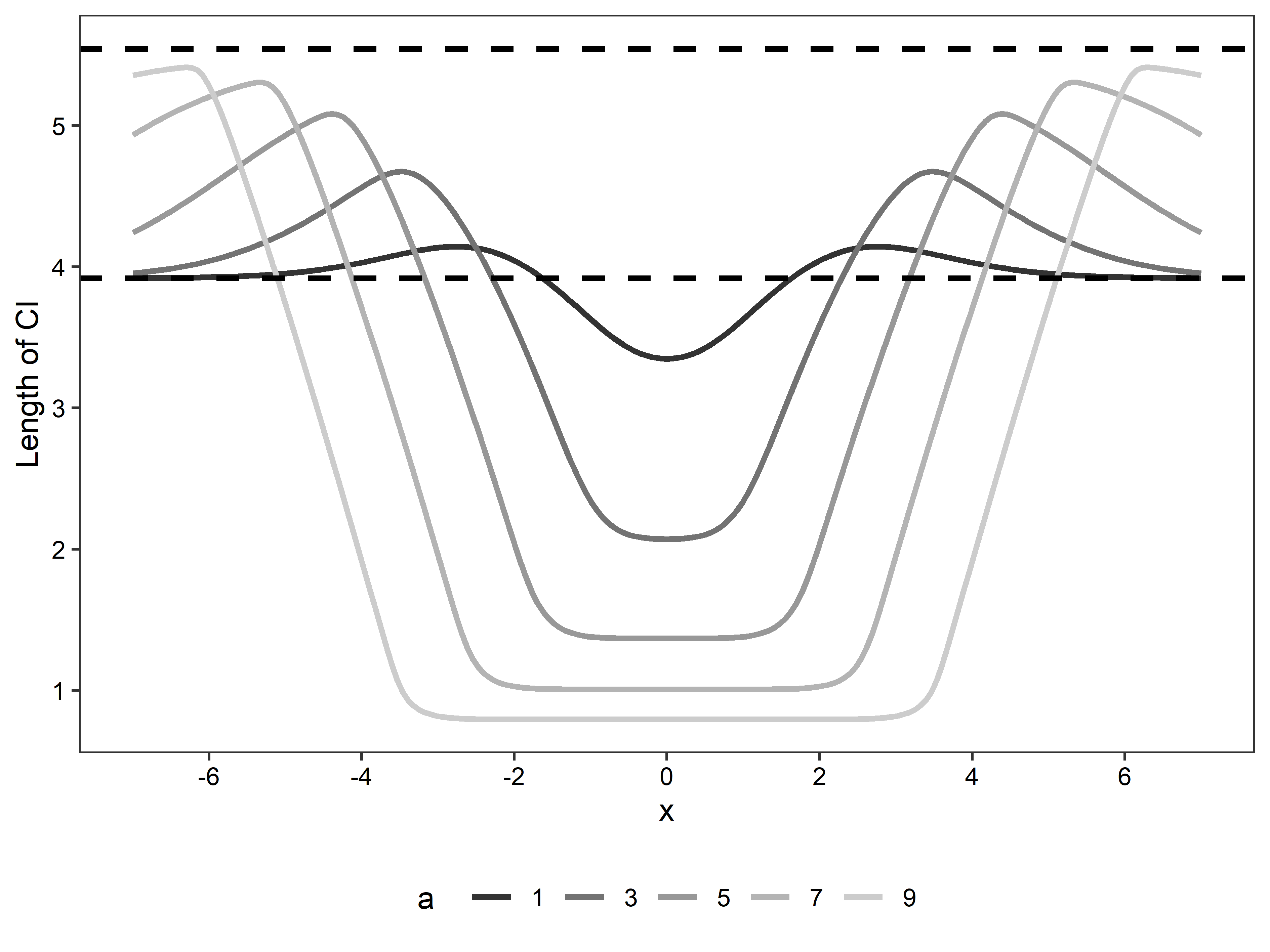}}
\end{minipage}
\caption{The length $[\mu_{\alpha_1}(z),\mu_{\alpha_2}(z)]$ is plotted as a function of 
$z$. In the left panel the truncation sets are of the form $(-a, a)$ and in the 
right panel truncation sets are of the form $(-\infty, -a) \cup (a,\infty)$. 
The different values for $a$ are shown below the plot. The remaining 
parameters are $\alpha_1=1-\alpha_2=0.025$ and $\sigma^2 = \tau^2 = 1$.}
\label{fig_length}
\end{figure}

In Figure~\ref{fig_expected_length}, we plot Monte-Carlo approximations of
the conditional expected length 
of $[\mu_{\alpha_1}(Z),\mu_{\alpha_2}(Z)]$ given $Z+R\in T$ as a function of $\mu$,
for the same scenarii as considered in Figure~\ref{fig_length}.
For the Monte-Carlo simulations, we draw $2000$ independent 
samples  from the distribution in 
\eqref{cond_rv}, compute the confidence interval for each 
and estimate the conditional expected length by the sample mean 
of the lengths. 
In the left panel, we observe that the conditional expected length 
is minimized at $\mu=0$ and converges to the upper bound as $\mu$ 
diverges. We also note that the smaller the truncation set, the larger the 
conditional expected length.  This means that 
the conditional expected length is 
close to the upper bound if the probability of the conditioning event is small.
In the right panel, we again observe that the conditional expected length is 
minimized at $\mu=0$. However, it seems to converge not the upper bound 
but to $\sigma\left(\Phi^{-1}(\alpha_2) - \Phi^{-1}(\alpha_1) \right)$ as $\mu$ diverges. 
Surprisingly, the conditional expected length at $\mu=0$ decreases as 
$a$ increases and becomes significantly smaller than 
$\sigma\left(\Phi^{-1}(\alpha_2) - \Phi^{-1}(\alpha_1) \right)$.  
In particular, 
and in contrast to the left panel, 
we here observe that, for $\mu$ close to zero, 
the conditional expected length decreases as 
the probability of the conditioning event decreases; for $\mu$ not close
to zero, the situation is again as in the left panel.

\begin{figure} 
\centering
\begin{minipage}{.5\textwidth}
  \centering
  \centerline{\includegraphics[trim=0 10 0 0, clip, 
width=1\textwidth]{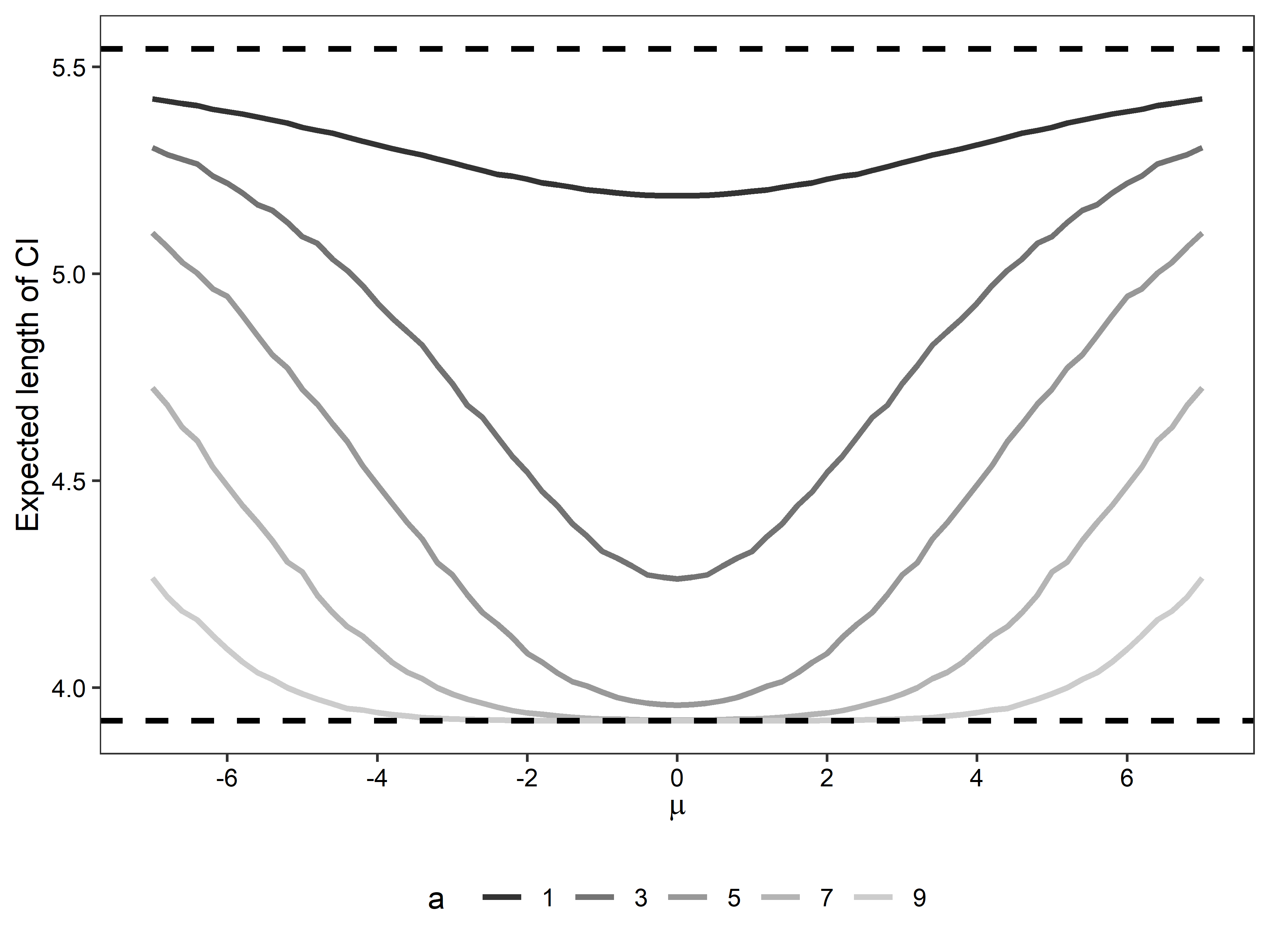}}
\end{minipage}%
\begin{minipage}{.5\textwidth}
  \centering
  \centerline{\includegraphics[trim=0 10 0 0, clip, 
width=1\textwidth]{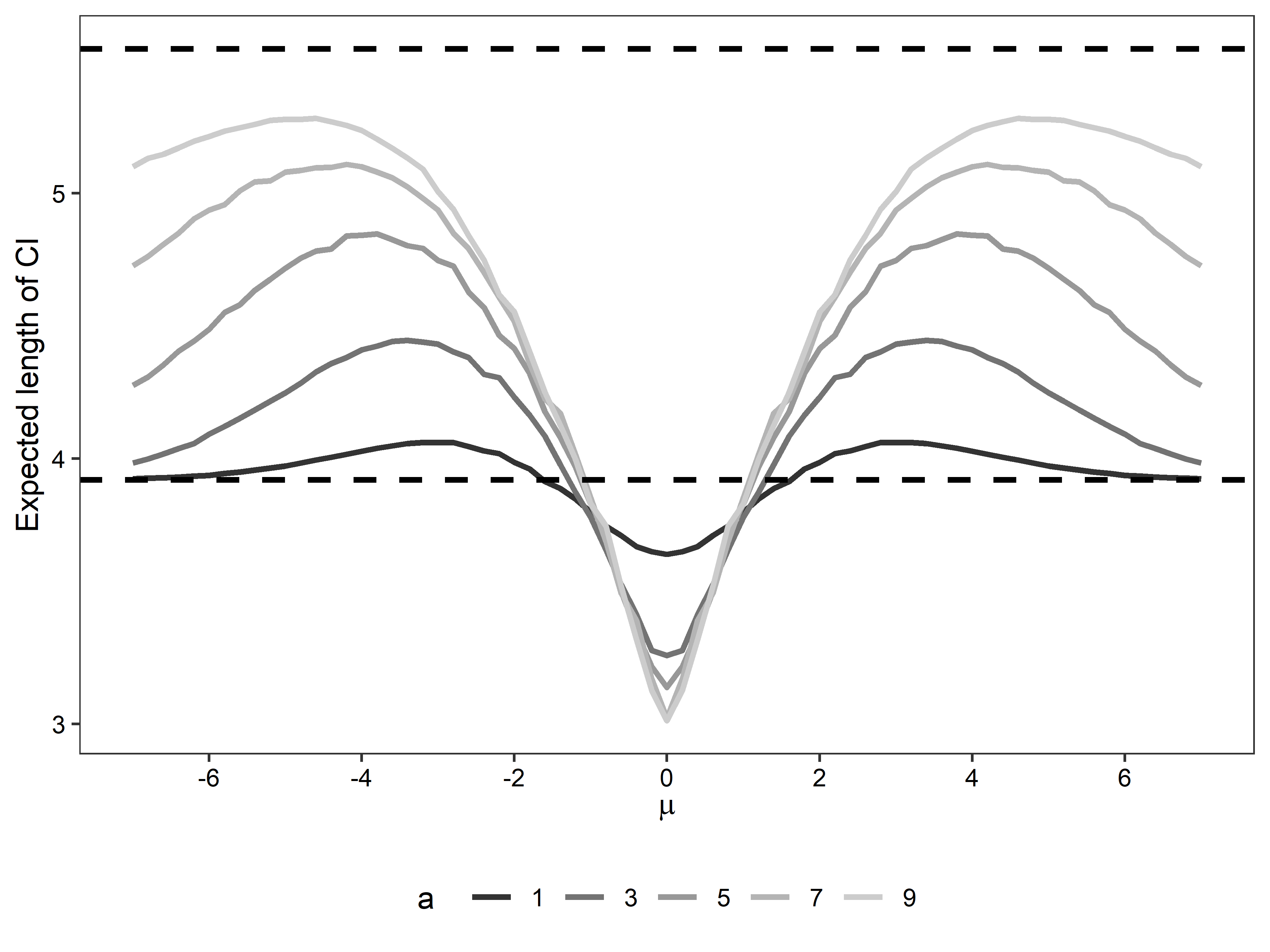}}
\end{minipage}
\caption{The conditional expected length of $[\mu_{\alpha_1}(Z),\mu_{\alpha_2}(Z)]$ is 
plotted as a function of $\mu$. In the left panel the truncation sets are of 
the form $(-a, a)$ and in the right panel the truncation sets are of the form 
$(-\infty, -a) \cup (a,\infty)$. The different values for $a$ are shown below 
the plot. The remaining parameters are $\alpha_1=1-\alpha_2=0.025$ and $\sigma^2 = 
\tau^2 = 1$.
}
\label{fig_expected_length}
\end{figure}

\begin{remark}[The unknown-variance case]  \normalfont \label{remark_variance}
In the discussion so far, the variances $\sigma^2$ and $\tau^2$ were assumed
to be known. Assume now that one or both of these variances are unknown,
and that variance estimators 
$\hat{\sigma}^2$ and $\hat{\tau}^2$ are available that take 
values in $(0,\infty)$;
if one of the variances is
known, set the corresponding estimator equal to its value.
We note that the truncation set $T$ as in \eqref{trunc_set}
may depend on the (estimated)
variances, and we stress this dependence here by denoting it by $\hat{T}$.
A natural way to obtain a confidence interval in the unknown variance case
is to proceed as before, using the variance estimators as plug-ins.
In particular, following the construction leading up to Theorem~\ref{th_core},
with $\sigma^2$, $\tau^2$ and $T$ replaced by $\hat{\sigma}^2$, 
$\hat{\tau}^2$ and $\hat{T}$, respectively, we obtain a confidence
interval for $\mu$
that we denote by $[\hat{\mu}_{\alpha_1}(Z), \hat{\mu}_{\alpha_2}(Z)]$.
For this interval, a relation like \eqref{eq_coverage} typically does not hold,
and its conditional coverage probability depends on the estimators
$\hat{\sigma}^2$ and $\hat{\tau}^2$; this topic is further discussed
in Section 8.1 of \cite{lee2016}. 
However, Theorem~\ref{th_core} can
still be used to obtain an upper bound on the length of this interval:
Using the theorem with $\sigma^2$, $\tau^2$ and $T$
replaced by $\hat{\sigma}^2$, $\hat{\tau}^2$ and $\hat{T}$, respectively,
we see that $\hat{\mu}_{\alpha_2}(Z) - \hat{\mu}_{\alpha_1}(Z)$ is smaller than
$\hat{\sigma}(\Phi^{-1}(\alpha_2) - 
\Phi^{-1}(\alpha_1))\sqrt{1+\hat{\sigma}^2/\hat{\tau}^2}$.
This is because the theorem only requires that \eqref{eq_bound} holds
with $\alpha$ replaced by $\alpha_i$, $i=1,2$.
\end{remark}

\section{Application to selective inference} \label{sec_selective}

Throughout this section,
we consider the generic sample mean setting that lies at the heart
of the polyhedral method and of many procedures derived from it. 
To use these methods with specific model selectors,
one essentially has to reduce the specific situation at hand
to the generic setting considered here. 
This is demonstrated by the examples in Section~\ref{sec_examples};
further examples can be obtained from the papers on selective inference
mentioned in Section~\ref{sec_intro}.
For the sake of exposition, we focus here on the known-variance case. 
As outlined in Remark~\ref{remark_variance}, our results can also be applied
in situations where variances are estimated, mutatis mutandis.

Let $n\in\mathbb N$ and let $Z_1,\dots, Z_n$ be i.i.d. normal random variables
with mean $\mu\in \mathbb R$ and variance $\sigma^2\in (0,\infty)$. The
outcome of a model-selection procedure can often be characterized through
an event of the form $\bar Z_n \in T$, where $\bar Z_n$ denotes the sample mean
and $T$ is as in \eqref{trunc_set}.
The polyhedral method provides a confidence interval for $\mu$
with pre-specified coverage 
probability conditional on the event $\bar Z_n \in T$,
based on the conditional distribution of 
$\bar Z_n | \bar Z_n \in T$.
The (conditional) expected length of this interval is infinite if and only if
$T$ is bounded from above or from below; cf. \cite{kivaranovic2020}.

\subsection{Data carving}
\label{carving}

Data carving \citep{fithian2017}
means that only a subset of the data is used for model selection 
while the entire dataset is used
for inference based on the selected model. 
Let $\delta \in (0,1)$ be such that $\delta n$ is a positive integer.
If only the first $\delta n$ observations are used for selection,
the outcome of a model-selection procedure can often be characterized through
an event of the form $\bar Z_{\delta n} \in T$; here $\bar Z_{\delta n}$ is
the sample mean of the first $\delta n$ observations and the truncation set
$T$ is as in \eqref{trunc_set}. (Of course, the truncation sets
used by the plain polyhedral method and by the polyhedral method with
data carving might differ.)
Inference for $\mu$ is now based on the conditional 
distribution 
\begin{equation} \label{rv_carving}
    \bar Z_n ~ | ~ \bar Z_{\delta n} \in T.
\end{equation}
In the preceding display,
the conditioning variable $\bar Z_{\delta n}$ 
can be written as
$\bar Z_{\delta n}= \bar Z_n + \bar R$ for
$\bar R = \bar Z_{\delta n} - \bar Z_n$.
Using elementary properties of the normal distribution, 
it is easy to see that
$\bar{Z}_{n}$ and $\bar R$ are independent. 
We thus obtain the following:

\begin{proposition}
  The conditional c.d.f. of the random variable in \eqref{rv_carving} is equal 
to $F_{\mu,\tilde \sigma^2}^{\tilde \tau^2}(z)$ with truncation set $T$,
\begin{equation*}
    \tilde \sigma^2 ~ = ~ \frac{\sigma^2}{n} \quad \text{and} \quad \tilde 
\tau^2 ~ = ~ \frac{\sigma^2}{n}\frac{1-\delta}{\delta}.
\end{equation*}
\end{proposition}

Let $\tilde \mu_\alpha(z)$ 
satisfy $F_{\tilde \mu_\alpha(z),\tilde \sigma^2}^{\tilde \tau^2}(z) = 1-q$, where $\tilde 
\sigma^2$ and $\tilde \tau^2$ are as in the proposition.
Then $[\tilde \mu_{\alpha_1}(\bar Z_n), \tilde \mu_{\alpha_2}(\bar Z_n)]$ is
a confidence interval for $\mu$ with conditional coverage probability
$\alpha_2-\alpha_1$ given $\bar Z_{\delta n} \in T$ ($0<\alpha_1< \alpha_2<1)$.
Theorem~\ref{th_core} 
implies that
\begin{equation*}
    \tilde \mu_{\alpha_2}(z) - \tilde \mu_{\alpha_1}(z) ~ < ~ 
\frac{\sigma}{\sqrt{n}} \left(\Phi^{-1}(\alpha_2) - \Phi^{-1}(\alpha_1) \right)
\frac{1}{\sqrt{1-\delta}}.
\end{equation*}
We see that the length of $[\tilde 
\mu_{\alpha_1}(\bar Z_n),\tilde \mu_{\alpha_2}(\bar Z_n)]$ shrinks at the same 
$\sqrt{n}$-rate as in the unconditional case. The price of conditioning is at 
most the factor $1/\sqrt{1-\delta}$. Note that, for $\delta=1$,
data carving reduces to the polyhedral method.

A corresponding sample-splitting method is the following:
Again, the model is selected based on the first $\delta n$ observations,
resulting in the same event $\bar Z_{\delta n} \in T$. For subsequent
inference, however, only the last $(1-\delta) n$ observations are used.
Because these are independent of $\bar Z_{\delta n}$, one obtains
the standard confidence interval for $\mu$ based on 
the last $(1-\delta)n$ observations,
whose length is
$\sigma (\Phi^{-1}(\alpha_2) - \Phi^{-1}(\alpha_1)) /\sqrt{(1-\delta)n}$.
By the inequality in the preceding display, this 
sample-splitting interval is always
strictly larger than the interval obtained with data carving.

\subsection{Selection with a randomized response}
\label{randomized}

This method of \cite{tian2018} performs model-selection with 
a randomized version of the data (i.e., after adding noise),
while inference based on the selected model is performed with
the original data. Let $\omega \sim N(0,\tau^2 I_n)$ be a noise vector
independent of $Z_1,\dots, Z_n$ and write $\bar \omega_n$ for the mean
of its components. If, in the model-selection step, the randomized
data $Z_1+\omega_1,\dots Z_n + \omega_n$ are used instead of the original data,
then the outcome of the model-selection step can often be characterized
through an event of the form $\bar Z_n + \bar \omega_n \in T$ where
$T$ again is a truncation set as in Section~\ref{sec_core} (possibly different
from the truncation sets used by the plain polyhedral method or by 
the polyhedral method with data carving).
Here, inference for $\mu$ is based on the conditional distribution 
\begin{equation} \label{rv_lasso}
\bar Z_n ~|~\bar Z_n+\bar \omega_n \in T,
\end{equation} 
which is easy to compute.

\begin{proposition}
The conditional c.d.f. of the random variable in \eqref{rv_lasso} is 
equal to $F_{\mu, \bar \sigma^2}^{\bar \tau^2}(z)$ with truncation set 
$T$,
\begin{equation*}
    \bar \sigma^2 ~ = ~ \frac{\sigma^2}{n}
    \quad \text{and}\quad \bar \tau^2 ~ = ~ \frac{\tau^2}{n}.
\end{equation*}
\end{proposition}

Let $\bar \mu_\alpha(z)$ satisfy 
$F_{\bar \mu_\alpha(z), \bar \sigma^2}^{\bar \tau^2}(z) = 1-q$, 
where $\bar \sigma^2$ and 
$\bar \tau^2$ are as in the proposition, so that 
$[\bar\mu_{\alpha_1}(\bar Z_n), \bar\mu_{\alpha_2}(\bar Z_n)]$ is a confidence interval
for $\mu$ with conditional coverage probability $\alpha_2-\alpha_1$ given
$\bar Z_n+\bar\omega_n\in T$ ($0<\alpha_1<\alpha_2<1$).
With Theorem~\ref{th_core}, we see that
\begin{equation*}
\bar\mu_{\alpha_1}(z) - \bar\mu_{\alpha_2}(z) ~<~
\frac{\sigma}{\sqrt{n}} \left( \Phi^{-1}(\alpha_2) - \Phi^{-1}(\alpha_1)\right)
\sqrt{1+\frac{\sigma^2}{\tau^2}}.
\end{equation*}
Similarly to data carving,
the length of the interval shrinks at the same $\sqrt{n}$-rate as
in the unconditional case, where the price of conditioning is controlled
by the factor $\sqrt{1+\sigma^2/\tau^2}$, and the method reduces to
the polyhedral method if $\tau=0$.

To obtain a sample-splitting method that is comparable to selection with
a randomized response, we proceed as follows: We use the first 
$n \sigma^2/(\sigma^2+\tau^2)$ observations for model selection and the
remaining $m = n \tau^2/(\sigma^2+\tau^2)$ observations for inference
(assuming, for simplicity, that these numbers are positive integers).
Write $\tilde Z_m$ for the mean of the last $m$ observations. Because the
first $n-m$ observations are independent of $\tilde Z_m$, we thus obtain the
standard confidence interval for $\mu$ based on $\tilde Z_m$ 
with length $\sigma(\Phi^{-1}(\alpha_2) - \Phi^{-1}(\alpha_1))/\sqrt{m} =
(\sigma/\sqrt{n})(\Phi^{-1}(\alpha_2) - \Phi^{-1}(\alpha_1))\sqrt{1+\sigma^2/\tau^2}$.
In terms of length, this interval is always dominated by 
$[\bar \mu_{\alpha_1} (\bar Z_n) - \bar \mu_{\alpha_2}(\bar Z_n)]$ considered above.

For data carving, choosing a corresponding sample splitting method
was obvious; cf. Subsection~\ref{carving}. This is not the case
in the setting considered here.  Nevertheless, the considerations
in the preceding paragraph show that selection with a randomized response
dominates any sample splitting scheme that uses at most
$m=n \tau^2 /(\sigma^2+\tau^2)$ observation for inference and the rest for
selection.

\begin{remark} \normalfont
Throughout this section, we have considered Gaussian data.
In non-Gaussian settings, our results can be applied asymptotically,
provided that, in \eqref{rv_carving} or \eqref{rv_lasso},
(i) the estimator used in the inference step, i.e.,  $\bar{Z}_n$,
as well as the random variables in the conditioning event
are asymptotically jointly normal
and (ii) the probability of the conditioning event
does not vanish.
\end{remark}

\section{Examples}
\label{sec_examples}

\subsection{Sequential testing with data carving}
\label{subsec_seq}

Consider a situation where an experiment $Z$ is to be repeated independently
$n$ times in order to determine the mean $\mu$ of $Z$. However, the
whole process is to be stopped at an earlier stage 
if results do not look promising, 
e.g., in a simple clinical trial.
In particular, the process is to be stopped if the mean 
$\bar Z_{\delta n}$ of the first $\delta n$ repetitions fails to exceed
a certain threshold $c$ (assuming that $\delta$ is a pre-determined fraction
so that $\delta n$ is an integer less than $n$).
In this situation, a confidence interval for $\mu$ is desired conditional
on the event that $\bar Z_{\delta n} > c$, that is, in the event that
the process was not stopped early.
If $Z$ is assumed to be Gaussian with mean $\mu$ and variance $\sigma^2$,
this situation can be handled using the results of Section~\ref{carving}:
Set $T = (c,\infty)$.
Based on the conditional distribution of $\bar Z_n$ given 
$\bar Z_{\delta n} \in T$, one obtains an equal-tailed
confidence interval for $\mu$ with coverage probability $1-\alpha$
whose length is less than
$\frac{2 \sigma}{\sqrt{n}}\Phi^{-1}(1-\alpha/2) /\sqrt{1-\delta}$.
In particular, the early stopping rule, i.e., conditioning on 
$\bar Z_{\delta n} \in T$, results in intervals that are longer than
the standard interval (that is constructed without early stopping, i.e.,
without conditioning)
by a factor of less than $1/\sqrt{1-\delta}$.

In this example, we did not consider controlling for explanatory variables
for simplicity;
allowing for additional explanatory variables to influence the experimental
outcome is more complex and will be studied elsewhere.
We also note that, in the setting of the present example,
data carving controls a conditional probability that is not commonly
considered in group sequential testing. 
For the substantial body of literature in that area,
we refer to \cite{Jen00a}.

\subsection{Lasso selection with a randomized response}
\label{subsec_lasso}

To complete the discussion of the randomized
Lasso from Section~\ref{sec_overview} we will use 
the assumptions and the notation maintained there.
Considered a model $m$ and a sign-vector $s$ so that
$\mathbb P(\hat{m}(Y)=m,\hat{s}(Y)=s) >0$.
We already know that
the conditional distribution of $\eta_m'Y$ given
$\hat{m}(Y+\omega)=m$, $\hat{s}(Y+\omega)=s$ and $(I_n-P_{\eta_m})Y=w$
is of the form \eqref{cond_rv}
with $Z\sim N(\check{\mu},\check{\sigma}^2)$
and $R\sim N(0,\check{\tau}^2)$,
where $\check{\mu} = \eta_m'\theta$,
$\check{\sigma}^2 = \sigma^2 \|\eta_m\|^2$,
$\check{\tau}^2 = \tau^2 \|\eta_m\|^2$,
and where $T = T_{m,s}(w)$ is an interval.
Choose $\alpha_i$, $i=1,2$, satisfying $0<\alpha_1<\alpha_2 < 1$
and choose $\check{\mu}_{\alpha_i(z)}$ so that 
$F^{\check{\tau}^2}_{\check{\mu}_{\alpha_i}(z),\check{\sigma}^2}(z) =
1-\alpha_i$, $i=1,2$, where the c.d.f. is computed with the truncation
set $T_{m,s}(w)$ replacing $T$.
Clearly, $\check{\mu}_{\alpha_i}(z)$
depends on $z$, on $m$ and $s$, and on $w$ (through the truncation set).
Set $\check{L}_{m,s,w}(z) = \check{\mu}_{\alpha_1}(z)$ and
$\check{U}_{m,s,w}(z) = \check{\mu}_{\alpha_2}(z)$.
Then, by construction,
$$
\mathbb P( \eta_m'\theta \in [\check{L}_{m,s,w}(\eta_m'Y), 
\check{U}_{m,s,w}(\eta_m'Y)]
\,|\, \hat{m}(Y+\omega) = m, \hat{s}(Y+\omega)=s,
(I_n -P_{\eta_m})(Y+\omega) = w) \quad=\quad \alpha_2 - \alpha_1
$$
and, by Theorem~\ref{th_core},
$$
\check{U}_{m,s,w}(\eta_m'Y)- \check{L}_{m,s,w}(\eta_m'Y) \quad<\quad
\check{\sigma} (\Phi^{-2}(\alpha_2) -\Phi^{-2}(\alpha_1))
\sqrt{ 1+\frac{\sigma^2}{\tau^2}}.
$$
Also, the relations in the preceding two displays continue to hold
if the conditioning on $(I_n-P_{\eta_m})(Y+\omega)=w$ is dropped in the first
display and if $w$ is replaced by $W = (I_n-P_{\eta_m})(Y+\omega)$ in both.

Now consider a model $m$ with $\mathbb P(\hat{m}(Y)) > 0$.
If $m=\emptyset$, then the event $\{\hat{m}(Y)=m\}$ is
again a polyhedron in $Y$-space; cf. \cite{lee2016}.
If $m\neq \emptyset$, then the event $\{\hat{m}(Y)=m\}$
can be decomposed into the (disjoint) union of
events of the form $\{\hat{m}(Y)=m,\hat{s}(Y)=s\}$,
each with positive probability;
therefore, the event $\{\hat{m}(Y)=m\}$
is the (disjoint) union of finitely many polyhedra in $Y$-space.
This entails that the conditional distribution of $\eta_m'Y$
given $\hat{m}(Y+\omega) = m$  and $(I_n-P_{\eta_m})(Y+\omega) = w$
is again of the form \eqref{cond_rv} with $Z$ and $R$ as in the
preceding paragraph, where $T = T_m(w)$ is now the union
of finitely many intervals. Now, proceeding as in the 
preceding paragraph, we obtain a confidence interval
$[\check{L}_{m,w}(\eta_m'Y), \check{U}_{m,w}(\eta_m'Y)]$
that satisfies
$$
\mathbb P( \eta_m'\theta \in [\check{L}_{m,w}(\eta_m'Y), \check{U}_{m,w}(\eta_m'Y)]
\,|\, \hat{m}(Y+\omega) = m, 
(I_n -P_{\eta_m})(Y+\omega) = w) \quad=\quad \alpha_2 - \alpha_1
$$
by construction and
$$
\check{U}_{m,w}(\eta_m'Y)- \check{L}_{m,w}(\eta_m'Y) \quad<\quad
\check{\sigma} (\Phi^{-2}(\alpha_2) -\Phi^{-2}(\alpha_1))
\sqrt{ 1+\frac{\sigma^2}{\tau^2}}
$$
by Theorem~\ref{th_core}.
As before, the relations in the preceding two displays continue to hold
if the conditioning on $(I_n-P_{\eta_m})(Y+\omega)=w$ is dropped in the first
display and if $w$ is replaced by $W = (I_n-P_{\eta_m})(Y+\omega)$ in both.

\section{Discussion} \label{sec_discussion}

We have shown that the length of certain confidence intervals
with conditional coverage guarantee can be drastically shortened
by adding some noise to the data throughout the model-selection
step. Examples include
data carving and selection on a randomized response, both combined
with the polyhedral method. 
Our findings clearly support the observations of
\cite{fithian2017} and \cite{tian2018} that sacrificing some
power in the model-selection step results
in an significant increase in power in subsequent inferences.
Selection and inference on the same data is not favorable in the case where
the events describing the outcome of the selection step correspond to
bounded regions in sample space (in our case, the truncation set $T$),
because then the resulting confidence set has infinite expected length;
cf. \cite{kivaranovic2020}.
There are, however, situations where this case can not occur:
For example, \cite{heller2019} study a situation where first a 
global hypothesis is tested against a two-sided alternative and subsequent 
tests are only performed if the global hypothesis is rejected.
There, bounded 
selection regions do not arise and excessively long intervals are not an issue.
Hence, we recommend to be cautious about the selection procedure 
one chooses. In some situations, adding noise in the selection step 
(e.g., through data carving or randomized selection)
may be beneficial; in other situations, it may not be necessary.

\begin{appendices}

\section{Proof of Theorem~\ref{th_core}}

We first provide some 
intuition behind the theorem. Second, we state Proposition~\ref{prop_core1} 
and~\ref{prop_core2} which are the two core results which the proof of 
Theorem~\ref{th_core} relies on. Finally, we prove Theorem~\ref{th_core} 
with the help of 
these two propositions. In Section~\ref{sec_prop_core1} 
and~\ref{sec_prop_core2} we then prove  Proposition~\ref{prop_core1} 
and~\ref{prop_core2}, respectively. The first of these two propositions is 
considerably more difficult to prove. In Section~\ref{sec_aux} we 
collect several auxiliary results which are required for the proofs of the 
main results.

Throughout this section, fix $\sigma^2$ and $\tau^2$ in $(0,\infty)$, and
simplify notation by setting $F_\mu (z) = F_{\mu,\sigma^2}^{\tau^2}(z)$ and
$f_\mu (z) = f_{\mu,\sigma^2}^{\tau^2}(z)$,
where $F_{\mu,\sigma^2}^{\tau^2}(z)$ and  $f_{\mu,\sigma^2}^{\tau^2}(z)$ 
denote the conditional c.d.f. and the conditional probability density
function (p.d.f.), respectively,
of the random variable in \eqref{cond_rv}. 
Set $\rho^2 = \tau^2 / (\sigma^2+\tau^2)$, and recall that
$\mu_\alpha(z)$ is defined by 
\eqref{eq_bound}. Denote by $\phi(t)$ and $\Phi(t)$ 
the p.d.f. and c.d.f. of the standard normal distribution with the 
usual convention that 
$\phi(-\infty)=\phi(\infty)=\Phi(-\infty)=1-\Phi(\infty)=0$.

Observe that the random vector $(Z, Z+R)'$ has a two-dimensional normal 
distribution with mean $(\mu,\mu)'$, variance $(\sigma^2,\sigma^2+\tau^2)'$ and 
covariance $\sigma^2$. It is elementary to verify that, for any 
$v \in \mathbb R$,
\begin{equation} \label{rv_cond_v}
    Z ~ | ~ Z + R = v ~ \sim ~ N(\rho^2\mu + (1-\rho^2)v, \sigma^2\rho^2).
\end{equation}
Let $G_\mu(z, v)$ denote the c.d.f. of this normal distribution, i.e.,
\begin{equation} \label{cdf_normal}
    G_\mu(z,v) = \Phi\left(\frac{z - (\rho^2\mu + (1-\rho^2)v)}{\sigma\rho} 
\right).
\end{equation}
By definition of $F_\mu(z)$, we have
\begin{equation} \label{eq_cdf_cond}
    F_\mu(z) = \mathbb E\left[G_\mu(z, V_\mu) \right],
\end{equation}
where $V_\mu$ is a random variable that is truncated normally distributed with 
mean $\mu$, variance $\sigma^2+\tau^2$ and truncation set $T$. 

Assume, for this paragraph, that $T$ is the singleton set $T = \{v\}$
for some fixed $v\in\mathbb R$ (singleton truncation sets are
excluded by our definition of $T$ in \eqref{trunc_set}).
Then the c.d.f.s $G_\mu(z,v)$ and $F_\mu(z)$ coincide, and it is
elementary to verify that the length of 
$[\mu_{\alpha_1}(z),\mu_{\alpha_2}(z)]$ is equal to $(\sigma/\rho)(\Phi^{-1}(\alpha_2) 
-\Phi^{-1}(\alpha_1))$, which is exactly the upper bound in 
Theorem~\ref{th_core}. The theorem thus implies that confidence 
intervals only become shorter if one conditions on a set $T$ with positive 
Lebesgue measure instead of a singleton. On the other hand, if $T$ is equal to 
$\mathbb R$, it is clear that the length of $[\mu_{\alpha_1}(z),\mu_{\alpha_2}(z)]$ is 
equal to $\sigma(\Phi^{-1}(\alpha_2)-\Phi^{-1}(\alpha_1))$. Surprisingly, this 
latter quantity is not 
necessarily a lower bound 
for the length of $[\mu_{\alpha_1}(z),\mu_{\alpha_2}(z)]$ if $T$ is
a proper subset of $\mathbb R$; cf.
the r.h.s. of Figure~\ref{fig_length}, or in Figure 1 of \cite{kivaranovic2020} in the case where $\tau$ is equal to $0$. 

\begin{proposition} \label{prop_core1}
  For all $z \in \mathbb R$ and all $\mu \in \mathbb R$, we have 
  \begin{equation} \label{ineq_dmu}
    \frac{\partial \Phi^{-1}(F_\mu(z))}{\partial \mu} ~ < ~ 
-\frac{\rho}{\sigma}.
  \end{equation}
\end{proposition}


\begin{proposition} \label{prop_core2}
Let $G_\mu(z,v)$ be defined as in \eqref{cdf_normal} and let $\alpha \in (0,1)$. 
If $\sup T = b_k < \infty$, then
    \begin{equation*}
        \lim_{z \to \infty} G_{\mu_{\alpha}(z)}(z,b_k) = 1-\alpha.
    \end{equation*}
    If $\inf T = a_1 > -\infty$, then
    \begin{equation*}
        \lim_{z \to -\infty} G_{\mu_{\alpha}(z)}(z,a_1) = 1-\alpha.
    \end{equation*}
\end{proposition}

This proposition entails 
that $G_{\mu_{\alpha}(z)}(z,b_k)$  converges to
$F_{\mu_{\alpha}(z)}(z)$ 
as $z \to \infty$ if the truncation set $T$ is 
bounded from above, and 
the same is true with $a_1$ replacing $b_k$  and as
$z\to-\infty$ if $T$ is bounded from below.
We continue now with the proof of Theorem~\ref{th_core}.

\begin{proof}[Proof of Theorem~\ref{th_core}]
By definition of $\mu_{\alpha_1}(z)$, we have 
$\Phi^{-1}(F_{\mu_{\alpha_1}(z)}(z)) = \Phi^{-1}(1-\alpha_1)$. 
Because Proposition~\ref{prop_core1} holds for any $z$ and 
$\mu$, it follows that for any $c \in (0,\infty)$, we have
  \begin{equation*}
    \Phi^{-1}(F_{\mu_{\alpha_1}(z) + c}(z)) < \Phi^{-1}(1-\alpha_1) 
    - c \rho /\sigma.
  \end{equation*}
We plug $c=\sigma(\Phi^{-1}(\alpha_2) - \Phi^{-1}(\alpha_1))/\rho$ 
into the inequality, 
apply the strictly increasing function $\Phi$ to both sides and use the 
symmetry $\Phi^{-1}(\alpha) = -\Phi^{-1}(1-\alpha)$ to obtain
  \begin{equation*}
    F_{\mu_{\alpha_1}(z) + \sigma(\Phi^{-1}(\alpha_2) - 
    \Phi^{-1}(\alpha_1))/\rho}(z) < 1-\alpha_2.
  \end{equation*}
  Because $F_\mu(z)$ is strictly decreasing in $\mu$ by Lemma~\ref{le_cdf} and 
$\mu_{\alpha_2}(z)$ satisfies the equation $F_{\mu_{\alpha_2}(z)}(z) = 1-\alpha_2$ , the 
previous inequality implies that
  \begin{equation*}
    \mu_{\alpha_2}(z) < \mu_{\alpha_1}(z) + \frac{\sigma}{\rho} \left(\Phi^{-1}(\alpha_2) - 
\Phi^{-1}(\alpha_1) \right).
  \end{equation*}
  Subtracting $\mu_{\alpha_1}(z)$ on both sides gives the inequality of the theorem. 
It remains to show that this upper bound is tight if the truncation set $T$ is 
bounded. We only consider the case $\sup T = b_k < \infty$ here, because
the case $\inf T = a_1 > -\infty$ can be treated by similar arguments, 
mutatis mutandis.
In view of the definition of 
$G_\mu(z,v)$ in \eqref{cdf_normal}, Proposition~\ref{prop_core2} and the 
symmetry $\Phi^{-1}(\alpha) = -\Phi^{-1}(1-\alpha)$ imply that
  \begin{equation*}
    \lim_{z \to \infty} \frac{z - (\rho^2 \mu_{\alpha_1}(z) + 
(1-\rho^2)b_k)}{\sigma\rho} = -\Phi^{-1}(\alpha_1)
  \end{equation*}
  and
  \begin{equation*} 
  \lim_{z \to \infty} \frac{z - (\rho^2 \mu_{\alpha_2}(z) + 
(1-\rho^2)b_k)}{\sigma\rho} = -\Phi^{-1}(\alpha_2).
  \end{equation*}
  Subtracting the second limit from the first and multiplying by 
$\sigma/\rho$ gives
    \begin{equation*}
      \lim_{z\to\infty} \mu_{\alpha_2}(z) - \mu_{\alpha_1}(z) = \frac{\sigma}{\rho} 
(\Phi^{-1}(\alpha_2)-\Phi^{-1}(\alpha_1)).
  \end{equation*}
  Hence the upper bound is tight.
\end{proof}

\subsection{Proof of Proposition~\ref{prop_core1}} \label{sec_prop_core1}
The proof of the Proposition is split up into a sequence of lemmas that are 
directly proven here.
\begin{lemma} \label{le_dx}
    For all $z \in \mathbb R$ and all $\mu \in \mathbb R$, we have 
    \begin{equation} \label{ineq_dx}
        \frac{\partial \Phi^{-1}(F_\mu(z))}{\partial z} ~ < ~ 
\frac{1}{\sigma\rho}.
    \end{equation}
\end{lemma}
\begin{proof}
    By the inverse function theorem, we have $\partial \Phi^{-1}(\alpha) / \partial 
\alpha = 1/\phi(\Phi^{-1}(\alpha))$. This equation and the chain rule imply that
\begin{equation*}
    \frac{\partial \Phi^{-1}(F_\mu(z))}{\partial z} = 
\frac{f_\mu(z)}{\phi(\Phi^{-1}(F_\mu(z)))}.
\end{equation*}
It suffices to show that the numerator on the r.h.s. is less than
$\phi(\Phi^{-1}(F_\mu(z))) / (\sigma\rho)$. 

In view of 
\eqref{cdf_normal}--\eqref{eq_cdf_cond}, Leibniz's rule 
implies that the p.d.f. $f_\mu(z)$ is equal to
\begin{equation*}
    f_\mu(z) ~ = ~ \frac{1}{\sigma\rho} \mathbb E\left[\phi\left( 
\frac{z-(\rho^2\mu + (1-\rho^2)V_\mu)}{\sigma\rho}\right) \right] ~ = 
\frac{1}{\sigma\rho} \mathbb 
E\left[\phi\left(\Phi^{-1}(G_\mu(z,V_\mu))\right)\right].
\end{equation*}
Observe now that it is sufficient to show that the function 
$\phi(\Phi^{-1}(\alpha))$ is strictly concave because, by Jensen's inequality, it 
follows then that $f_\mu(z)$ is bounded from above by
\begin{equation*}
    \frac{1}{\sigma\rho} \phi\left(\Phi^{-1}\left(\mathbb 
E\left[G_\mu(z,V_\mu)\right]\right)\right) ~ = ~ 
\frac{1}{\sigma\rho}\phi(\Phi^{-1}(F_\mu(z))),
\end{equation*}
completing the proof.

Elementary calculus shows that 
\begin{equation*}
    \frac{ \partial^2 \phi(\Phi^{-1}(\alpha))}{\partial \alpha^2} = 
\frac{-1}{\phi(\Phi^{-1}(\alpha))}.
\end{equation*}
Because $\phi(\Phi^{-1}(\alpha))$ is positive for all $\alpha \in (0,1)$, it follows that 
the second derivative is negative for all $\alpha \in (0,1)$. This means, 
$\phi(\Phi^{-1}(\alpha))$ is strictly concave and the proof is complete.
\end{proof}
Note that the inequality \eqref{ineq_dx} of this lemma resembles inequality 
\eqref{ineq_dmu} of Proposition~\ref{prop_core1}. While inequality 
\eqref{ineq_dx} is surprisingly easy to prove, inequality \eqref{ineq_dmu} is 
more difficult. Equation \eqref{eq_cdf_cond} provides intuition why this is the 
case: The distribution of the random variable $V_\mu$ does not depend on $z$ 
but it depends on $\mu$. Hence to prove inequality \eqref{ineq_dmu}, we cannot 
exchange integral and differential and we cannot apply Jensen's inequality as 
we did in the proof of Lemma~\ref{le_dx}. We did not find a direct proof of 
inequality \eqref{ineq_dmu}. However, in the following we show that inequality 
\eqref{ineq_dx} in fact implies \eqref{ineq_dmu}.

To show this implication, we need a more explicit representation of $f_\mu(z)$. 
Elementary calculus and properties of the conditional normal distribution imply 
that the conditional p.d.f. $f_\mu(z)$ can also be written as
\begin{equation} \label{eq_pdf_v2}
  f_\mu(z) ~ = ~ 
\frac{1}{\sigma}\phi\left(\frac{z-\mu}{\sigma}\right)\frac{\sum_{i=1}^k 
\Phi\left(\frac{b_i-z}{\tau}\right) - \Phi\left(\frac{a_i-z}{\tau}\right) 
}{\sum_{i=1}^k \Phi\left(\frac{b_i - \mu}{\sqrt{\sigma^2+\tau^2}}\right) - 
\Phi\left(\frac{a_i - \mu}{\sqrt{\sigma^2+\tau^2}}\right)}.
\end{equation}

\begin{lemma} \label{le_dmu_formula}
    Let $G_\mu(z, v)$ be defined as in \eqref{cdf_normal}. For all $z \in 
\mathbb R$ and all $\mu \in \mathbb R$, we have 
  \begin{equation*}
    \frac{\partial \Phi^{-1}(F_\mu(z))}{\partial \mu} = \frac{-f_\mu(z) + 
\sum_{i=1}^k 
h(b_i)(F_\mu(z)-G_\mu(z,b_i))-h(a_i)(F_\mu(z)-G_\mu(z,a_i))}{
 \phi(\Phi^{-1}(F_\mu(z)))},
  \end{equation*}
  where 
  \begin{equation*}
    h(v) = \frac{\frac{1}{\sqrt{\sigma^2+\tau^2}}\phi\left(\frac{v - 
\mu}{\sqrt{\sigma^2+\tau^2}}\right)}{\sum_{i=1}^k \Phi\left(\frac{b_i - 
\mu}{\sqrt{\sigma^2+\tau^2}}\right) - \Phi\left(\frac{a_i - 
\mu}{\sqrt{\sigma^2+\tau^2}}\right)}.
  \end{equation*}
\end{lemma}
\begin{proof}
    The chain rule implies that
  \begin{equation*}
    \frac{\partial \Phi^{-1}(F_\mu(z))}{\partial \mu} = \frac{\partial 
\Phi^{-1}(F_\mu(z))}{\partial F_\mu(z)} \frac{\partial F_\mu(z)}{\partial \mu}.
  \end{equation*}
  The inverse function theorem implies that the first derivative on the r.h.s. 
is equal to $1/\phi(\Phi^{-1}(F_\mu(z)))$. Therefore, it remains to show that 
$\partial F_\mu(z)/\partial \mu$
is equal to the numerator on the r.h.s. of the equation 
of the lemma. Leibniz's rule implies that $\partial F_\mu(z)/ \partial \mu = 
\int_{-\infty}^z \partial f_\mu(u) / \partial\mu ~ du$. Therefore, we compute 
$\partial f_\mu(z) / \partial\mu$ first. Lemma~\ref{le_dens_dmu} implies that
  \begin{equation*}
    \frac{\partial f_\mu(z)}{\partial \mu} = \frac{z-\mu}{\sigma^2}f_\mu(z) + 
f_\mu(z) \sum_{i=1}^k h(b_i) - h(a_i).
  \end{equation*}
We use the expression on the r.h.s. to compute
$\int_{-\infty}^x \partial f_\mu(u) /\partial \mu d u$.
    Lemma~\ref{le_int_f1} implies that the integral of the first summand is 
equal to
  \begin{equation*}
    -f_\mu(z) - \sum_{i=1}^k h(b_i)G_\mu(z,b_i) - h(a_i)G_\mu(z,a_i).
  \end{equation*}
  Because, in the second-to-last display, 
  the second summand on the r.h.s. depends on $x$ only through 
$f_\mu(z)$, it is easy to see the integral of this function is equal to
  \begin{equation*}
      F_\mu(z) \sum_{i=1}^k h(b_i) - h(a_i).
  \end{equation*}
   The sum of the last two expressions is equal to the numerator on the r.h.s of
   the equation of the lemma, which 
completes the proof.
\end{proof}

This lemma implies that proving Proposition~\ref{prop_core1} is equivalent to 
showing that $B_\mu(z)<0$ for
\begin{multline} \label{eq_B}
    B_\mu(z) = \frac{\rho}{\sigma} \phi(\Phi^{-1}(F_\mu(z))) -f_\mu(z) \\ + 
\sum_{i=1}^k h(b_i)(F_\mu(z)-G_\mu(z,b_i))-h(a_i)(F_\mu(z)-G_\mu(z,a_i)).
\end{multline}
Observe that
\begin{align*}
    0 &= \lim_{|z|\to\infty} f_\mu(z) = \lim_{|z|\to\infty} 
\phi(\Phi^{-1}(F_\mu(z)) \\
    &= \lim_{z\to-\infty} F_\mu(z) =  \lim_{z\to\infty} 1-F_\mu(z) \\
    &= \lim_{z\to-\infty} G_\mu(z,v) =  \lim_{z\to\infty} 1-G_\mu(z,v).
\end{align*}
This equation chain implies that $B_\mu(z)$ converges to $0$ as $|z| \to 
\infty$. Holding $\mu, \sigma^2$ and $\tau^2$ fixed, this is the same as saying 
that $B_\mu(z)$ converges to $0$ as $F_\mu(z)$ goes to $0$ or $1$. Let the 
function $F_\mu^{-1}(\alpha)$ be defined by the equation
\begin{equation} \label{F_inv}
    F_\mu(F_\mu^{-1}(\alpha)) = \alpha.
\end{equation}
Clearly, $F_\mu^{-1}(\alpha)$ is well-defined for all $\alpha \in (0,1)$ and we have that 
$F_\mu^{-1}(F_\mu(z)) = z$ for all $z \in \mathbb R$. To prove 
Proposition~\ref{prop_core1}, it is now sufficient to show that 
$B_\mu(F_\mu^{-1}(\alpha))$ is 
strictly convex as a function of $\alpha$ for any fixed $\mu, \sigma^2$ 
and $\tau^2$ (in view of the second-to-last display).

\begin{lemma}
    Let $B_\mu(z)$ be defined in \eqref{eq_B} and $F_\mu^{-1}(\alpha)$ in 
\eqref{F_inv}. Let $\mu, \sigma^2$ and $\tau^2$ be fixed. Then, for all $z \in 
\mathbb R$, 
    \begin{equation*}
        \frac{\partial^2 B_\mu(F_\mu^{-1}(\alpha))}{\partial \alpha^2}\Bigr\rvert_{\alpha = 
F_\mu(z)} = -\frac{\rho}{\sigma \phi(\Phi^{-1}(F_\mu(z)))} + \frac{1}{\sigma^2 
f_\mu(z)}.
    \end{equation*}
\end{lemma}
\begin{proof}
We start by computing the first derivative of 
$B_\mu(F_\mu^{-1}(\alpha))$ with 
respect to $\alpha$. The chain rule implies that
    \begin{equation*}
        \frac{\partial B_\mu(F_\mu^{-1}(\alpha))}{\partial \alpha} =  
	\frac{\partial 
B_\mu(F_\mu^{-1}(\alpha))}{\partial F_\mu^{-1}(\alpha)} \frac{\partial 
F_\mu^{-1}(\alpha)}{\partial \alpha}.
    \end{equation*}
    The inverse function theorem implies that the second derivative on the 
r.h.s. is equal to $1/ f_\mu(F_\mu^{-1}(\alpha))$. In view of the 
definitions of 
$B_\mu(z)$ and $F_\mu^{-1}(\alpha)$ in \eqref{eq_B} and \eqref{F_inv}, we 
see that to compute the first derivative on the r.h.s. we need to compute the 
derivatives of $\alpha$, $\phi(\Phi^{-1}(\alpha))$, $f_\mu(F_\mu^{-1}(\alpha))$ and 
$G_\mu(F_\mu^{-1}(\alpha),v)$ with respect to $F_\mu^{-1}(\alpha)$. 
Because $\partial 
F_\mu^{-1}(\alpha) / \partial \alpha$ is equal to 
$1/ f_\mu(F_\mu^{-1}(\alpha))$, it follows 
that the derivative of $\alpha$ with respect to $F_\mu^{-1}(\alpha)$ 
is equal to 
$f_\mu(F_\mu^{-1}(\alpha))$. The chain rule implies that the derivative of 
$\phi(\Phi^{-1}(\alpha))$ with respect to $F_\mu^{-1}(\alpha)$ is equal to 
$-\Phi^{-1}(\alpha) 
f_\mu(F_\mu^{-1}(\alpha))$. Lemma~\ref{le_dens_dx} implies that the 
derivative of 
$f_\mu(F_\mu^{-1}(\alpha))$ with respect to $F_\mu^{-1}(\alpha)$ is equal to
    \begin{equation*}
        -\frac{F_\mu^{-1}(\alpha)-\mu}{\sigma^2} f_\mu(F_\mu^{-1}(\alpha)) - 
f_\mu(F_\mu^{-1}(\alpha)) \sum_{i=1}^k l(F_\mu^{-1}(\alpha),b_i) - 
l(F_\mu^{-1}(\alpha),a_i),
    \end{equation*}
where $l(z,v)$ is defined in Lemma~\ref{le_dens_dx}. And finally, 
Lemma~\ref{le_dens_G_dx} implies that the derivative of 
$G_\mu(F_\mu^{-1}(\alpha),v)$ with 
respect to $F_\mu^{-1}(\alpha)$ is equal to
    \begin{equation*}
        \frac{f_\mu(F_\mu^{-1}(\alpha)) l(F_\mu^{-1}(\alpha),v)}{h(v)},
    \end{equation*}
    where $h(v)$ is defined in Lemma~\ref{le_dmu_formula} and $l(z,v)$ is 
defined in Lemma~\ref{le_dens_dx}. The previous four derivatives and the 
definition of $B_\mu(z)$ in \eqref{eq_B} entail, after 
straight-forward simplifications, that
    \begin{equation*}
        \frac{\partial B_\mu(F_\mu^{-1}(\alpha))}{\partial F_\mu^{-1}(\alpha)} = 
f_\mu(F_\mu^{-1}(\alpha)) \left( -\frac{\rho}{\sigma} \Phi^{-1}(\alpha) + 
\frac{F_\mu^{-1}(\alpha)-\mu}{\sigma^2} + \sum_{i=1}^k h(b_i) - h(a_i) \right)
    \end{equation*}
    and therefore
    \begin{equation*}
        \frac{\partial B_\mu(F_\mu^{-1}(\alpha))}{\partial \alpha} = -\frac{\rho}{\sigma} 
\Phi^{-1}(\alpha) + \frac{F_\mu^{-1}(\alpha)-\mu}{\sigma^2} + \sum_{i=1}^k h(b_i) - 
h(a_i).
    \end{equation*} 
    Now it is easy to see that
    \begin{equation*}
        \frac{\partial^2 B_\mu(F_\mu^{-1}(\alpha))}{\partial \alpha^2} = 
-\frac{\rho}{\sigma \phi(\Phi^{-1}(\alpha))} + \frac{1}{\sigma^2 
f_\mu(F_\mu^{-1}(\alpha))}.
    \end{equation*}
    The claim of the lemma follows by evaluating the second derivative at 
$\alpha=F_\mu(z)$. 
\end{proof}
To prove Proposition~\ref{prop_core1}, it remains to show that 
\begin{equation*}
    -\frac{\rho}{\sigma \phi(\Phi^{-1}(F_\mu(z)))} + \frac{1}{\sigma^2 
f_\mu(z)} > 0.
\end{equation*}
But this is the same as showing
\begin{equation*}
    \frac{f_\mu(z)}{\phi(\Phi^{-1}(F_\mu(z)))} < \frac{1}{\sigma \rho}.
\end{equation*}
The l.h.s. is equal to $\partial \Phi^{-1}(F_\mu(z)) / \partial z$ and in 
Lemma~\ref{le_dx} we have shown
that this inequality is true. This completes the proof of 
Proposition~\ref{prop_core1}.

\subsection{Proof of Proposition~\ref{prop_core2}} \label{sec_prop_core2}

Let $\epsilon > 0$. We will only consider the case where $\sup T = b_k < 
\infty$; the other case follows by similar arguments, mutatis mutandis. 
Recall the definition of 
$F_\mu(z)$ in \eqref{eq_cdf_cond}. By the law of total probability, we can 
write $F_\mu(z)$ as
\begin{equation*}
    \mathbb P(V_\mu >b_k -\epsilon) \mathbb E\left[G_\mu(z, V_\mu)| V_\mu > 
b_k-\epsilon \right] + \mathbb P(V_\mu \leq b_k -\epsilon) \mathbb 
E\left[G_\mu(z, V_\mu)| V_\mu \leq b_k-\epsilon \right].
\end{equation*}
Note that both conditional expectations are bounded by $1$. Because the random 
variable $V_\mu$ (defined under equation \eqref{eq_cdf_cond}) is a truncated 
normal with mean $\mu$, variance $\sigma^2+\tau^2$ and truncation set $T$, it 
follows that $V_\mu$ converges in probability to $b_k$ as $\mu$ goes to 
$\infty$. Now as $z$ goes to $\infty$, it follows that $\mu_\alpha(z)$ goes to 
$\infty$ (cf. the discussion after the proof of Lemma~\ref{le_cdf}). 
This implies that
\begin{equation*}
    \lim_{z\to\infty} F_{\mu_{\alpha}(z)}(z) = \lim_{z\to\infty} \mathbb 
E\left[G_{\mu_{\alpha}(z)}(z, V_{\mu_{\alpha}(z)})| V_{\mu_{\alpha}(z)} > b_k-\epsilon \right].
\end{equation*}
Note that $G_\mu(z,v)$ is strictly decreasing in $v$. This means that 
$F_\mu(z)$ is bounded from below by $G_\mu(z,b_k)$. But this also means that 
$\lim_{z\to\infty} F_{\mu_{\alpha}(z)}(z)$ is bounded from below by 
$\lim_{z\to\infty} G_{\mu_{\alpha}(z)}(z,b_k)$. On the other hand, 
observe that the 
conditional expectation on the r.h.s. of the preceding display is bounded 
from above by $G_{\mu_\alpha(z)}(z,b_k-\epsilon)$. 
This means that $\lim_{z\to\infty} 
F_{\mu_{\alpha}(z)}(z)$ is bounded from above by $\lim_{z\to\infty} 
G_{\mu_{\alpha}(z)}(z,b_k- \epsilon)$. Since $F_{\mu_{\alpha}(z)}(z)$ is 
equal to $1-\alpha$ 
for all $z \in \mathbb R$, it follows that
\begin{equation*}
    \lim_{z\to\infty} G_{\mu_{\alpha}(z)}(z,b_k) \leq 1-\alpha 
    \leq \lim_{z\to\infty} 
G_{\mu_{\alpha}(z)}(z,b_k- \epsilon).
\end{equation*}
Because $\epsilon$ was arbitrary and $G_\mu(z,v)$ is simply a normal c.d.f.,
which is uniformly continuous,
the claim of the proposition follows.

\subsection{Auxiliary results} \label{sec_aux}

\begin{lemma} \label{le_cdf}
  For every $z \in \mathbb R$, $F_\mu(z)$ is continuous and strictly decreasing 
in $\mu$ and satisfies
  \begin{equation*}
      \lim_{\mu \to \infty} F_\mu(z) = \lim_{\mu \to -\infty} 1-F_\mu(z) = 0.
  \end{equation*}
\end{lemma}
\begin{proof}
    Continuity is obvious. For monotonicity, it is sufficient to show that 
$f_\mu(z)$ has monotone likelihood ratio because \cite{lee2016} already showed 
that monotone likelihood ratio implies monotonicity. This means for $\mu_1 < 
\mu_2$, we need to show that $f_{\mu_2}(z) / f_{\mu_1}(z)$ is strictly 
increasing in $z$. In view of the definition of $f_\mu(z)$ in 
\eqref{eq_pdf_v2}, it is easy to see that $f_{\mu_2}(z) / f_{\mu_1}(z)$ can be 
written as $c \exp((\mu_2-\mu_1)x/\sigma)$, where $c$ is a positive constant 
that does not depend on $z$. Because $\mu_1 < \mu_2$ and the exponential 
function is strictly increasing, it follows that $f_\mu(z)$ has monotone 
likelihood ratio. Finally, we show that $\lim_{\mu \to \infty} F_\mu(z) = 0$. 
The other part of this equation follows by similar arguments. Let $M < b_k$. 
Recall the definition of $F_\mu(z)$ in \eqref{eq_cdf_cond}. By the law of total 
probability,  we can write $F_\mu(z)$ as
    \begin{equation*}
        \mathbb P(V_\mu > M) \mathbb E\left[G_\mu(z, V_\mu)| V_\mu > M \right] 
+ \mathbb P(V_\mu \leq M) \mathbb E\left[G_\mu(z, V_\mu)| V_\mu \leq M \right].
    \end{equation*}
    Note that both conditional expectations are bounded by $1$. Because the 
random variable $V_\mu$ (defined under equation \eqref{eq_cdf_cond}) is a 
truncated normal with mean $\mu$, variance $\sigma^2+\tau^2$ and truncation set 
$T$, it follows that $\lim_{\mu\to\infty} \mathbb P(V_\mu > M) = 1- 
\lim_{\mu\to\infty} \mathbb P(V_\mu \leq M) = 1$. This implies that
    \begin{equation*}
        \lim_{\mu\to\infty} F_\mu(z) = \lim_{\mu\to\infty} \mathbb 
E\left[G_\mu(z, V_\mu)| V_\mu > M \right].
    \end{equation*}
    Because $G_\mu(z,v)$ is strictly decreasing in $v$, it follows that the 
conditional expectation on the r.h.s. is bounded from above by $G_\mu(z, M)$. 
But this means that $\lim_{\mu\to\infty} F_\mu(z)$ is bounded by 
$\lim_{\mu\to\infty} G_\mu(z, M)$. Because latter limit is equal to $0$, the 
same follows for the former limit.
\end{proof}
This lemma ensures that the function $\mu_\alpha(z)$ is well defined, continuous, 
strictly increasing in $z$ and $\alpha$ and that $\lim_{z\to\infty} \mu_\alpha(z) = 
\lim_{z\to-\infty} -\mu_\alpha(z)=\infty$ (see also
Lemma A.3 in \cite{kivaranovic2020}).

\begin{lemma} \label{le_dens_dmu}
    Let the function $h(v)$ be defined as in Lemma~\ref{le_dmu_formula}. For 
all $x \in \mathbb R$ and all $\mu \in \mathbb R$, we have
    \begin{equation*}
        \frac{\partial f_\mu(z)}{\partial \mu} = \frac{x-\mu}{\sigma^2}f_\mu(z) 
+ f_\mu(z) \sum_{i=1}^k h(b_i) - h(a_i).
    \end{equation*}
\end{lemma}
\begin{proof}
    In view of the definition of $f_\mu(z)$ in \eqref{eq_pdf_v2}, the chain 
rule and product rule imply that the derivative of $f_\mu(z)$ with respect to 
$\mu$ is equal to
    \begin{gather*}
\frac{x-\mu}{\sigma^2}\frac{1}{\sigma}\phi\left(\frac{x-\mu}{\sigma}\right)\frac
{\sum_{i=1}^k \Phi\left(\frac{b_i-x}{\tau}\right) - 
\Phi\left(\frac{a_i-x}{\tau}\right) }{\sum_{i=1}^k \Phi\left(\frac{b_i - 
\mu}{\sqrt{\sigma^2+\tau^2}}\right) - \Phi\left(\frac{a_i - 
\mu}{\sqrt{\sigma^2+\tau^2}}\right)} \\
        + 
\frac{1}{\sigma}\phi\left(\frac{x-\mu}{\sigma}\right)\frac{\left(\sum_{i=1}^k 
\Phi\left(\frac{b_i-x}{\tau}\right) - 
\Phi\left(\frac{a_i-x}{\tau}\right)\right) \left( 
\frac{1}{\sqrt{\sigma^2+\tau^2}} \sum_{i=1}^k \phi\left(\frac{b_i - 
\mu}{\sqrt{\sigma^2+\tau^2}}\right) - \phi\left(\frac{a_i - 
\mu}{\sqrt{\sigma^2+\tau^2}}\right) \right)}{ \left(\sum_{i=1}^k 
\Phi\left(\frac{b_i - \mu}{\sqrt{\sigma^2+\tau^2}}\right) - \Phi\left(\frac{a_i 
- \mu}{\sqrt{\sigma^2+\tau^2}}\right)\right)^2}.
    \end{gather*}
    In view of the definition of $f_\mu(z)$, it is easy to see that the first 
summand is equal to 
    \begin{equation*}
        \frac{x-\mu}{\sigma^2}f_\mu(z)
    \end{equation*}
    and, in view of the definitions of $f_\mu(z)$ and $h(v)$, that the second 
summand is equal to
    \begin{equation*}
        f_\mu(z) \sum_{i=1}^k h(b_i) - h(a_i).
    \end{equation*}
\end{proof}

\begin{lemma} \label{le_int_f1}
    Let $G_\mu(z, v)$ be defined as in \eqref{cdf_normal} and $h(v)$ as in 
Lemma~\ref{le_dmu_formula}. For all $x \in \mathbb R$ and all $\mu \in \mathbb 
R$, we have
    \begin{equation*}
        \int_{-\infty}^z \frac{u-\mu}{\sigma^2} f_\mu(u) du ~ = ~ -f_\mu(z) - 
\sum_{i=1}^k  h(b_i)G_\mu(z,b_i) - h(a_i)G_\mu(z,a_i).
    \end{equation*}
\end{lemma}
\begin{proof}
    By definition of $f_\mu(z)$ in \eqref{eq_pdf_v2}, the integral can be 
written as
    \begin{equation*}
        \frac{\sum_{i=1}^k \int_{-\infty}^z \frac{u-\mu}{\sigma^3} 
\phi\left(\frac{u-\mu}{\sigma}\right) \Phi\left(\frac{b_i-u}{\tau}\right) du - 
\int_{-\infty}^z \frac{u-\mu}{\sigma^3} \phi\left(\frac{u-\mu}{\sigma}\right) 
\Phi\left(\frac{a_i-u}{\tau}\right) du }{\sum_{i=1}^k \Phi\left(\frac{b_i - 
\mu}{\sqrt{\sigma^2+\tau^2}}\right) - \Phi\left(\frac{a_i - 
\mu}{\sqrt{\sigma^2+\tau^2}}\right)}.
    \end{equation*}
    Note that all integrands in the numerator are of the same form. They only 
differ in the constants $a_1,b_1,\dots,a_k,b_k$. This means, we can apply 
Equation 10,011.1 of \cite{owen1980} to each integral. This equation implies 
that the numerator is equal to
    \begin{align*}
        &\sum_{i=1}^k \frac{-1}{\sqrt{\sigma^2+\tau^2}} \phi\left(\frac{b_i - 
\mu}{\sqrt{\sigma^2+\tau^2}} \right) \Phi\left(\frac{z - (\rho^2\mu + 
(1-\rho^2)b_i)}{\sigma\rho} \right) - \frac{1}{\sigma}  
\phi\left(\frac{z-\mu}{\sigma}\right) \Phi\left(\frac{b_i-z}{\tau}\right) \\
        &- \left(\frac{-1}{\sqrt{\sigma^2+\tau^2}} \phi\left(\frac{a_i - 
\mu}{\sqrt{\sigma^2+\tau^2}} \right) \Phi\left(\frac{z - (\rho^2\mu + 
(1-\rho^2)a_i)}{\sigma\rho} \right) - \frac{1}{\sigma}  
\phi\left(\frac{z-\mu}{\sigma}\right) \Phi\left(\frac{a_i-z}{\tau}\right) 
\right).
    \end{align*}
    (Also note that this equation can easily be verified by differentiation of 
the antiderivative.) In view of the definitions of $f_\mu(z)$, $h(v)$ and 
$G_\mu(z, v)$, we can see that the claim of the lemma is true.
\end{proof}

\begin{lemma} \label{le_dens_dx}
    For all $z \in \mathbb R$ and all $\mu \in \mathbb R$, we have
    \begin{equation*}
        \frac{\partial f_\mu(z)}{\partial z} =  -\frac{z-\mu}{\sigma^2}f_\mu(z) 
- f_\mu(z) \sum_{i=1}^k l(z,b_i) - l(z,a_i),
    \end{equation*}
    where
    \begin{equation*}
        l(z,v) = \frac{\frac{1}{\tau} \phi\left(\frac{v 
-z}{\tau}\right)}{\sum_{i=1}^k \Phi\left(\frac{b_i -z}{\tau}\right) -  
\Phi\left(\frac{a_i -z}{\tau}\right)}.
    \end{equation*}
\end{lemma}
\begin{proof}
    In view of the definition of $f_\mu(z)$ in \eqref{eq_pdf_v2}, the chain 
rule and product rule imply that the derivative of $f_\mu(z)$ with respect to 
$z$ is equal to
    \begin{align*}
&-\frac{z-\mu}{\sigma^2}\frac{1}{\sigma}\phi\left(\frac{z-\mu}{\sigma}\right)
\frac{\sum_{i=1}^k \Phi\left(\frac{b_i-z}{\tau}\right) - 
\Phi\left(\frac{a_i-z}{\tau}\right) }{\sum_{i=1}^k \Phi\left(\frac{b_i - 
\mu}{\sqrt{\sigma^2+\tau^2}}\right) - \Phi\left(\frac{a_i - 
\mu}{\sqrt{\sigma^2+\tau^2}}\right)} \\
        &- 
\frac{1}{\sigma}\phi\left(\frac{z-\mu}{\sigma}\right)\frac{\frac{1}{\tau} 
\sum_{i=1}^k \phi\left(\frac{b_i-z}{\tau}\right) - 
\phi\left(\frac{a_i-z}{\tau}\right)}{ \sum_{i=1}^k \Phi\left(\frac{b_i - 
\mu}{\sqrt{\sigma^2+\tau^2}}\right) - \Phi\left(\frac{a_i - 
\mu}{\sqrt{\sigma^2+\tau^2}}\right) }.
    \end{align*}
    It is easy to see that the first summand is equal to 
    \begin{equation*}
        -\frac{z-\mu}{\sigma^2}f_\mu(z)
    \end{equation*}
    and, in view of the definitions of $f_\mu(z)$ and $l(z,v)$, that the second 
summand is equal to
    \begin{equation*}
        -f_\mu(z) \sum_{i=1}^k l(z, b_i) - l(z,a_i).
    \end{equation*}
\end{proof}

\begin{lemma} \label{le_dens_G_dx}
    Let $h(v)$ be defined as in Lemma~\ref{le_dmu_formula} and $l(z,v)$ as in 
Lemma~\ref{le_dens_dx}. For all $z \in \mathbb R$ and all $\mu \in \mathbb R$, 
we have
    \begin{equation*}
        \frac{\partial G_\mu(z,v)}{\partial z} = \frac{f_\mu(z) l(z,v)}{h(v)}.
    \end{equation*}    
\end{lemma}

\begin{proof}
    By definition of $G_\mu(z,v)$ in \eqref{cdf_normal}, we have
    \begin{equation*}
        \frac{\partial G_\mu(z,v)}{\partial z} = \frac{1}{\sigma\rho} 
\phi\left(\frac{z-(\rho^2\mu+ (1-\rho^2)v)}{\sigma\rho}\right).
    \end{equation*}
    We claim that     
    \begin{equation*}
        \frac{1}{\sigma\rho} \phi\left(\frac{x-(\rho^2\mu+ 
(1-\rho^2)v)}{\sigma\rho}\right) = 
\frac{\frac{1}{\sigma}\phi\left(\frac{x-\mu}{\sigma}\right)\frac{1}{\tau}\phi
\left(\frac{v-x}{\tau}\right)}{\frac{1}{\sqrt{\sigma^2+\tau^2}} 
\phi\left(\frac{v-\mu}{\sqrt{\sigma^2+\tau^2}}\right)}.
    \end{equation*}
    To see this note that we can write the l.h.s. as $c_1 \exp(-d_1/2)$ and 
r.h.s. as $c_2 \exp(-d_2/2)$, where
    \begin{equation*}
        c_1 = \frac{1}{\sigma\rho \sqrt{2\pi}} = 
\frac{\sqrt{\sigma^2+\tau^2}}{\sigma \tau \sqrt{2\pi}} = c_2
    \end{equation*}
    and
    \begin{align*}
        d_1 &= \left( \frac{z-(\rho^2\mu+ (1-\rho^2)v)}{\sigma\rho} \right)^2 \\
        &= \frac{z^2 - 2\rho^2\mu z - 2(1-\rho^2)v z + \rho^4\mu^2 + 
(1-\rho^2)^2 v^2 + 2\rho^2(1-\rho^2)\mu v}{\sigma^2\rho^2} \\
        &= \frac{\rho^2(z-\mu)^2 + (1-\rho^2)(v-z)^2 - 
\rho^2(1-\rho^2)(v-\mu)^2 }{\sigma^2\rho^2} \\
        &= \left(\frac{z-\mu}{\sigma}\right)^2 + \left( \frac{v-z}{\tau} 
\right)^2 - \left(\frac{v-\mu}{\sqrt{\sigma^2+\tau^2}}\right)^2 = d_2.
    \end{align*}
    Hence the claimed equation is true. Observe that the r.h.s. of that 
equation can be written as
    \begin{equation*}
        \Scale[1.05]{  
          \frac{1}{\sigma}\phi\left(\frac{z-\mu}{\sigma}\right) 
\frac{\sum_{i=1}^k \Phi\left(\frac{b_i -z}{\tau}\right) - \Phi\left(\frac{a_i 
-z}{\tau}\right)}{\sum_{i=1}^k \Phi\left(\frac{b_i - 
\mu}{\sqrt{\sigma^2+\tau^2}}\right) - \Phi\left(\frac{a_i - 
\mu}{\sqrt{\sigma^2+\tau^2}}\right)} 
\frac{\frac{1}{\tau}\phi\left(\frac{v-z}{\tau}\right)}{\sum_{i=1}^k 
\Phi\left(\frac{b_i -z}{\tau}\right) - \Phi\left(\frac{a_i -z}{\tau}\right)} 
\frac{\sum_{i=1}^k \Phi\left(\frac{b_i - \mu}{\sqrt{\sigma^2+\tau^2}}\right) - 
\Phi\left(\frac{a_i - 
\mu}{\sqrt{\sigma^2+\tau^2}}\right)}{\frac{1}{\sqrt{\sigma^2+\tau^2}} 
\phi\left(\frac{v-\mu}{\sqrt{\sigma^2+\tau^2}}\right)}      
        }.
    \end{equation*}
In view of the definitions of $f_\mu(z)$, 
$l(z,v)$ and $h(v)$, it is easy 
to see that the previous expression is equal to $f_\mu(z) l(z,v) / h(v)$. Hence 
the derivative of $G_\mu(z,v)$ with respect to $z$ is of the claimed form.
\end{proof}

\end{appendices}

\bibliographystyle{apalike}
\bibliography{bibliography}

\begin{thebibliography}{}

\bibitem[Bachoc et~al., 2019]{bachoc2019}
Bachoc, F., Leeb, H., and Pötscher, B.~M. (2019).
\newblock Valid confidence intervals for post-model-selection predictors.
\newblock {\em Annals of Statistics}, 47:1475--1504.

\bibitem[Bachoc et~al., 2020]{bachoc2020}
Bachoc, F., Preinerstorfer, D., and Steinberger, L. (2020).
\newblock Uniformly valid confidence intervals post-model-selection.
\newblock {\em Annals of Statistics}, 48:440--463.

\bibitem[Berk et~al., 2013]{berk2013}
Berk, R., Brown, L., Buja, A., Zhang, K., and Zhao, L. (2013).
\newblock Valid post-selection inference.
\newblock {\em Annals of Statistics}, 41:802--837.

\bibitem[Fithian et~al., 2017]{fithian2017}
Fithian, W., Sun, D.~L., and Taylor, J. (2017).
\newblock Optimal inference after model selection.
\newblock {\em arXiv preprint arxiv:1410.2597}.

\bibitem[Frank and Friedman, 1993]{frank1993}
Frank, I.~E. and Friedman, J.~H. (1993).
\newblock A statistical view of some chemometrics regression tools.
\newblock {\em Technometrics}, 35:109--135.

\bibitem[Heller et~al., 2019]{heller2019}
Heller, R., Meir, A., and Chatterjee, N. (2019).
\newblock Post-selection estimation and testing following aggregate association tests.
\newblock {\em Journal of the Royal Statistical Society: Series B (Statistical Methodology)}, 81:547--573.

\bibitem[Jennison and Turnbull, 2000]{Jen00a}
Jennison, C. and Turnbull, B.~W. (2000).
\newblock {\em Group sequential methods with applications to clinical trials}.
\newblock Chapman \& Hall/CRC, Boca Raton, FL.

\bibitem[Kivaranovic and Leeb, 2021]{kivaranovic2020}
Kivaranovic, D. and Leeb, H. (2021).
\newblock On the length of post-model-selection confidence intervals conditional on polyhedral constraints.
\newblock {\em Journal of the American Statistical Association}, 534:845--857.

\bibitem[Kuchibhotla et~al., 2018a]{kuchibhotla2018a}
Kuchibhotla, A.~K., Brown, L.~D., Buja, A., George, E.~I., and Zhao, L. (2018a).
\newblock A model free perspective for linear regression: Uniform-in-model bounds for post selection inference.
\newblock {\em arXiv preprint arXiv:1802.05801}.

\bibitem[Kuchibhotla et~al., 2018b]{kuchibhotla2018b}
Kuchibhotla, A.~K., Brown, L.~D., Buja, A., George, E.~I., and Zhao, L. (2018b).
\newblock Valid post-selection inference in assumption-lean linear regression.
\newblock {\em arXiv preprint arXiv:1806.04119}.

\bibitem[Lee et~al., 2016]{lee2016}
Lee, J.~D., Sun, D.~L., Sun, Y., and Taylor, J.~E. (2016).
\newblock Exact post-selection inference, with application to the lasso.
\newblock {\em Annals of Statistics}, 44:907--927.

\bibitem[Leeb and P\"otscher, 2005]{Lee03a}
Leeb, H. and P\"otscher, B.~M. (2005).
\newblock Model selection and inference: Facts and fiction.
\newblock {\em Econometric Theory}, {\bf 21}:21--59.

\bibitem[Leeb and P\"otscher, 2006]{leeb2006}
Leeb, H. and P\"otscher, B.~M. (2006).
\newblock Can one estimate the conditional distribution of post-model-selection estimators?
\newblock {\em Annals of Statistics}, 34:2554--2591.

\bibitem[Leeb and P\"otscher, 2008]{leeb2008}
Leeb, H. and P\"otscher, B.~M. (2008).
\newblock Can one estimate the unconditional distribution of post-model-selection estimators?
\newblock {\em Econometric Theory}, 24:338--376.

\bibitem[Lehmann and Romano, 2006]{lehmann2006}
Lehmann, E.~L. and Romano, J.~P. (2006).
\newblock {\em Testing statistical hypotheses}.
\newblock Springer Science \& Business Media.

\bibitem[Markovic et~al., 2018]{markovic2018}
Markovic, J., Xia, L., and Taylor, J. (2018).
\newblock Unifying approach to selective inference with applications to cross-validation.
\newblock {\em arXiv preprint arXiv:1703.06559}.

\bibitem[Owen, 1980]{owen1980}
Owen, D.~B. (1980).
\newblock A table of normal integrals.
\newblock {\em Communications in Statistics - Simulation and Computation}, 9:389--419.

\bibitem[Panigrahi and Taylor, 2019]{panigrahi2019}
Panigrahi, S. and Taylor, J. (2019).
\newblock Approximate selective inference via maximum-likelihood.
\newblock {\em arXiv preprint arXiv:1902.07884}.

\bibitem[Panigrahi et~al., 2018]{panigrahi2018}
Panigrahi, S., Zhu, J., and Sabatti, C. (2018).
\newblock Selection-adjusted inference: an application to confidence intervals for $cis$-e{QTL} effect sizes.
\newblock {\em arXiv preprint arXiv:1801.08686}.

\bibitem[Reid et~al., 2017]{reid2017}
Reid, S., Taylor, J., and Tibshirani, R. (2017).
\newblock Post-selection point and interval estimation of signal sizes in gaussian samples.
\newblock {\em Canadian Journal of Statistics}, 45:128--148.

\bibitem[Reid et~al., 2018]{reid2018}
Reid, S., Taylor, J., and Tibshirani, R. (2018).
\newblock A general framework for estimation and inference from clusters of features.
\newblock {\em Journal of the American Statistical Association}, 113:280--293.

\bibitem[Rosenthal, 1979]{Ros79a}
Rosenthal, R. (1979).
\newblock The ``{F}ile {D}rawer {P}roblem'' and tolerance for null results.
\newblock {\em Psychol. Bull.}, {\bf 86}:638--641.

\bibitem[Taylor and Tibshirani, 2018]{taylor2018}
Taylor, J. and Tibshirani, R. (2018).
\newblock Post-selection inference for $l_1$-penalized likelihood models.
\newblock {\em Canadian Journal of Statistics}, 46:41--61.

\bibitem[Tian et~al., 2018]{tian2018b}
Tian, X., Loftus, J.~R., and Taylor, J.~E. (2018).
\newblock Selective inference with unknown variance via the square-root lasso.
\newblock {\em Biometrika}, 105:755--768.

\bibitem[Tian et~al., 2016]{tian2016}
Tian, X., Panograhi, S. a. M.~J., Bi, N., and Taylor, J. (2016).
\newblock Selective sampling after solving a convex problem.
\newblock {\em arXiv preprint arXiv:1609.05609}.

\bibitem[Tian and Taylor, 2017]{tian2017}
Tian, X. and Taylor, J. (2017).
\newblock Asymptotics of selective inference.
\newblock {\em Scandinavian Journal of Statistics}, 44:480--499.

\bibitem[Tian and Taylor, 2018]{tian2018}
Tian, X. and Taylor, J. (2018).
\newblock Selective inference with a randomized response.
\newblock {\em Annals of Statistics}, 46:679--710.

\bibitem[Tibshirani, 1996]{tibshirani1996}
Tibshirani, R. (1996).
\newblock Regression shrinkage and selection via the lasso.
\newblock {\em Journal of the Royal Statistical Society. Series B (Methodological)}, 58:267--288.

\bibitem[Tibshirani, 2013]{tibshirani2013}
Tibshirani, R.~J. (2013).
\newblock The lasso problem and uniqueness.
\newblock {\em Electronic Journal of Statistics}, 7:1456--1490.

\bibitem[Tibshirani et~al., 2016]{tibshirani2016}
Tibshirani, R.~J., Taylor, J., Lockhart, R., and Tibshirani, R. (2016).
\newblock Exact post-selection inference for sequential regression procedures.
\newblock {\em Journal of the American Statistical Association}, 111:600--620.

\bibitem[Zrnic and Jordan, 2020]{zrnic2020}
Zrnic, T. and Jordan, M.~I. (2020).
\newblock Post-selection inference via algorithmic stability.
\newblock {\em arXiv preprint arXiv:2011.09462}.

\end{thebibliography}

\end{document}